\documentclass[11pt]{article}

\usepackage{physics}
\usepackage{braket}
\usepackage{mathtools}
\usepackage[margin=1in]{geometry}
\usepackage{amsmath,amssymb,amsfonts,setspace,url,mathrsfs}
\usepackage{bm}
\usepackage{algpseudocode}
\usepackage{enumitem}
\usepackage{latexsym}
\usepackage{amsthm}
\usepackage{subfigure}
\usepackage{float}
\usepackage{stmaryrd}
\usepackage{fancybox}
\usepackage{bigstrut,array,multirow}
\usepackage{here}
\usepackage{color}
\usepackage{url}
\usepackage{hyperref}
\usepackage[capitalise]{cleveref}
\usepackage{mathpazo}
\usepackage{thmtools}
\usepackage{thm-restate}
\usepackage[normalem]{ulem}

\usepackage{xcolor}
\usepackage{tikz}
\usetikzlibrary{arrows.meta}

\newtheorem{theorem}{Theorem}
\newtheorem{corollary}[theorem]{Corollary}

\newtheorem{lemma}[theorem]{Lemma}
\newtheorem{definition}[theorem]{Definition}

\newenvironment{properties}[2][0]
{
	\begin{enumerate} \setcounter{enumi}{#1}}{\end{enumerate}}

\newcommand{\dist}{\text{dist}}
\newcommand{\fset}{{\mathcal{F}}}
\newcommand{\qset}{{\mathcal{Q}^G}}
\newcommand{\pset}{{\mathcal{Q}^H}}
\newcommand{\dset}{{\mathcal{D}}}
\newcommand{\eps}{{\varepsilon}}
\newcommand{\poly}{\mathsf{poly}}
\newcommand{\congestion}{\textnormal{cong}}

\newcounter{note}
\newcommand{\zt}[1]{\refstepcounter{note}$\ll${\bf ZT~\thenote:}
	{\sf \color{red} #1}$\gg$\marginpar{\tiny\bf ZT~\thenote}}
\newcommand{\ymy}[1]{\refstepcounter{note}$\ll${\bf MY~\thenote:}
	{\sf \color{red} #1}$\gg$\marginpar{\tiny\bf MY~\thenote}}

\newcommand{\myparskip}{3pt}
\parskip \myparskip
\begin{document}

\begin{titlepage}
	
	\title{Lower Bounds on Flow Sparsifiers with Steiner Nodes}

\author{Yu Chen\thanks{National University of Singapore, Singapore. Email: {\tt yu.chen@nus.edu.sg}.} \and Zihan Tan\thanks{University of Minnesota Twin Cities, MN, USA. Email: {\tt ztan@umn.edu}.} \and Mingyang Yang\thanks{National University of Singapore, Singapore. Email: {\tt myangat@u.nus.edu}.}} 
	
	\maketitle

	\thispagestyle{empty}
	\begin{abstract}
Given a large graph $G$ with a set of its $k$ vertices called terminals, a \emph{quality-$q$ flow sparsifier} is a small graph $G'$ that contains the terminals and preserves all multicommodity flows between them up to some multiplicative factor $q\ge 1$, called the \emph{quality}. Constructing flow sparsifiers with good quality and small size ($|V(G')|$) has been a central problem in graph compression.

The most common approach of constructing flow sparsifiers is contraction: first compute a partition of the vertices in $V(G)$, and then contract each part into a supernode to obtain $G'$.
When $G'$ is only allowed to contain all terminals, the best quality is shown to be $O(\log k/\log\log k)$ 
and $\Omega(\sqrt{\log k/\log\log k})$. 
In this paper, we show that allowing a few Steiner nodes does not help much in improving the quality. Specifically, there exist $k$-terminal graphs such that, even if we allow $k\cdot 2^{(\log k)^{\Omega(1)}}$ Steiner nodes in its contraction-based flow sparsifier, the quality is still $\Omega\big((\log k)^{0.3}\big)$.

\end{abstract}
\end{titlepage}

\tableofcontents
\newpage

\section{Introduction}

Graph compression is a paradigm for transforming large graphs into smaller ones while retaining essential characteristics, such as flow/cut structures and shortest-path distances. By reducing graph size before further computations, it effectively saves computational resources. This approach has been instrumental in developing faster and more efficient approximation algorithms for graphs.

We study a specific type of graph compression called \emph{flow-approximating vertex sparsifiers\footnote{A closely related notion is cut-approximating vertex sparsifiers and will be discussed in \Cref{sec: related}.}}, first introduced in \cite{hagerup1998characterizing,moitra2009approximation,leighton2010extensions}.
In this setting, the input is a large graph $G$ and a subset $T\subseteq V(G)$ of $k$ special vertices called \emph{terminals}, and the goal is to compress $G$ into a small graph $G'$ that contains $T$ and preserves all multi-commodity flows between terminals in $T$, up to some factor $q\ge 1$ that is called the \emph{quality} of $G'$.\footnote{A formal definition is provided in \Cref{sec: prelim}.}
The focus in this research direction lies in the trade-off between the quality (the smaller $q$ is, the better) and the size (measured by the number of vertices in $G'$, the smaller the better) of flow sparsifiers.

In the restricted setting where $G'$ is required to contain only terminals,
Leighton and Moitra \cite{leighton2010extensions} showed that quality $O(\log k/\log\log k)$ is achievable, and \cite{charikar2010vertex}, \cite{makarychev2010metric} and \cite{englert2014vertex} showed that such sparsifiers can be constructed efficiently.
On the other hand, a lower bound of $\Omega(\log\log k)$ on quality is proved in \cite{leighton2010extensions} and improved to $\Omega(\sqrt{\log k/\log\log k})$ by Makarychev and Makarychev \cite{makarychev2016metric}, leaving a small gap yet to be closed.
Thus, the natural next question is:
\[\emph{Can better (in quality) sparsifiers be constructed by allowing a small number of Steiner vertices}? 
\]

This question has been a major open question in graph compression and has received much attention in recent years. For cut sparsifiers, there has been a long line of exciting work achieving significant progress \cite{hagerup1998characterizing,moitra2009approximation,chuzhoy2012vertex,kratsch2012representative,khan2014mimicking,karpov2019exponential,jambulapati2023sparsifying,chen2024cut}. For flow sparsifiers, the progress has been relatively slower. Chuzhoy \cite{chuzhoy2012vertex} showed constructions of flow sparsifiers of quality $O(1)$ and size $C^{O(\log\log C)}$, where $C$ is the total capacity of all terminal-incident edges (assuming each edge has capacity at least $1$). In addition, the quality-$(1+\eps)$ regime has been relatively well studied \cite{andoni2014towards,abraham2016fully,krauthgamer2023exact,chen20241+}. However, despite all the efforts, it is still unclear whether we can achieve quality $O(1)$ with $\tilde O(k)$ many Steiner nodes.

In addition to achieving both good quality and size, another desired property for a flow sparsifier $G'$ is that it inherits structural properties from the input graph $G$,\footnote{For example, if $G$ is a structurally simple graph, say a tree or a planar graph, we would also want $G'$ to be as simple.} or in other words, ``$G'$ should come from $G$''.
This leads us to the approach of contraction, the most common approach to constructing flow sparsifiers. Specifically, we first compute a partition $\fset$ of vertices in $G$, and then contract each set in $\fset$ into a supernode to obtain $G'$.
Such a graph $G'$ is called a \emph{contraction-based flow sparsifier} of $G$.
Contraction-based flow sparsifiers play a similar role in vertex sparsifiers as the role of spanners in distance-preserving edge sparsifiers, in that both allow only the use of edges from the input graph. To the best of our knowledge,
all previous constructions of flow sparsifiers are contraction-based flow sparsifiers or convex combinations of them \cite{leighton2010extensions,charikar2010vertex,makarychev2010metric,chuzhoy2012vertex,krauthgamer2013mimicking,andoni2014towards,englert2014vertex,chen20241+}.
Therefore, it is natural to explore whether we can achieve better quality (ideally $O(1)$) via 
contraction-based flow sparsifiers and their convex combinations, allowing a small number of Steiner nodes (say $\tilde O(k)$ many).

\subsection{Our results}

In this paper, we give a negative answer to this question. Our main result can be summarized as the following theorem.

\begin{restatable}{thm}{mainresult}
\label{thm: main}
For every large enough integer $k$ and every small constant $\eps>0$, there exists a graph $G$ with $k$ terminals, such that any convex combination of contraction-based flow sparsifiers on at most $k\cdot 2^{(\log k)^{(1-c_0)+\eps}}$ vertices has quality at least $(\log k)^{c_0-O(\eps)}$, where $c_0=\frac{1}{\log_2 5 + 1}\approx 0.301$ is a universal constant.
\end{restatable}


Our \Cref{thm: main} shows that, even if we allow a super-linear number of Steiner vertices ($k\cdot 2^{(\log k)^{\Omega(1)}}$ many, which is greater than $k\cdot\poly\log k$), contraction-based flow sparsifiers (and their convex combinations) still cannot achieve quality better than $\Omega\big((\log k)^{0.3}\big)$. Compared with the lower bound of $\tilde\Omega\big(\sqrt{\log k}\big)$ \cite{makarychev2016metric} in the case where no Steiner nodes are allowed, our result shows that adding a few Steiner nodes does not help much in improving its quality.

\textbf{Remark.} Slightly extending our approach, we can prove \Cref{thm: main} for every constant $c'_0 < 0.301$, as long as $\eps$ is small enough compared with $c'_0$, obtaining a tradeoff between quality and size controlled by $c'_0$. As the size bound will never go beyond $k\cdot 2^{(\log k)^{1-\Omega(1)}}$, which is always $k^{1+o(1)}$, and we tend to feature the strongest quality lower bound, we choose to formulate \Cref{thm: main} in the current way.

\subsection{Technical Overview}

We now provide some high-level intuition on the proof of \Cref{thm: main}. 

The hard instance is an expander with maximum degree bounded by a constant, and the terminals are an arbitrary subset of $k$ vertices.
The reason for choosing an expander is that, if a graph $G$ excludes a small clique $K_r$ as a minor, then the previous work \cite{leighton2010extensions}, via analyzing the solution cost of the $0$-Extension problem on graph $G$, has constructed a quality-$O(r^2)$ contraction-based flow sparsifier on terminals only (no Steiner nodes). Therefore, in order to prove a better quality lower bound even in the presence of Steiner vertices, our graph $G$ must be rich in structure, which makes expanders our favorable choice. 
We remark that expanders have also been used as the central object in previous lower bound proofs/attempts, for example, the $\Omega(\log\log k)$ lower bound in \cite{leighton2010extensions}, the $\Omega(\log k)$ lower bound on the integrality gap for the $0$-Extension problem \cite{calinescu2005approximation}, and the $\Omega(\log\log k)$ lower bound on the integrality gap for the LP relaxation of the $0$-Extension with Steiner Nodes problem in \cite{chen2024lower}. 
In this paper, we perform surgery on expanders and carry out a novel analysis which leads to a better bound.

In this section, we focus on providing an overview for the lower bound proof for contraction-based flow sparsifiers. The proof can be easily generalized to their convex combinations. See \Cref{sec: convex combination}.

Let $G$ be a constant-degree expander, let $T$ be a set of $k$ terminals which are chosen arbitrarily from $V(G)$ and let $H$ be a contraction-based flow sparsifier of $G$.
In order to prove a quality lower bound for $H$, we need to construct a demand $\dset$ on terminals, and show that $\dset$ can be routed via almost edge-disjoint paths in $H$ (causing congestion $O(1)$), but must incur congestion $(\log k)^{\Omega(1)}$ when routed in $G$. 
Therefore, our goal is to find a collection $\pset$ of (almost) edge-disjoint paths in $H$, such that 
\begin{itemize}
    \item each path $Q\in \pset$ connects a pair $s_Q,t_Q$ of  terminals in $T$; and \label{prop: terminal}
    \item the sum of all distances between endpoints in $G$ is large: $\sum_{Q\in \pset}\dist_G(s_Q,t_Q)=n\cdot (\log n)^{\Omega(1)}$. \label{prop: distance}
\end{itemize}
Note that the second property enforces that any routing of the demand $\dset^*=\set{(s_Q,t_Q)\mid Q\in \pset}$ in $G$ takes paths of total length at least $n\cdot (\log n)^{\Omega(1)}$, and since $G$ contains $O(n)$ edges (as the max degree of $G$ is $O(1)$), such a routing must incur congestion 
\begin{equation}
\label{eqn: cong}
\textnormal{cong}_G(\dset^*)\ge \frac{\sum_{Q\in \pset}\dist_G(s_Q,t_Q)}{|E(G)|}\ge 
\frac{n\cdot (\log n)^{\Omega(1)}}{O(n)}=(\log n)^{\Omega(1)}.
\end{equation}

\paragraph{Paths in $H$ and paths in $G$.}
Since we are more familiar with expanders themselves than their contracted graphs, analyzing paths in $G$ would be easier than arguing about paths in $H$. Moreover, there is a natural connection between paths in $G$ and paths in $H$: 
a path $(u_1,\ldots,u_r)$ in $G$ naturally induces a path $(F(u_1),\ldots,F(u_r))$ (possibly with repeated nodes) in $H$, where $F(u_i)$ represents the supernode in $H$ obtained by contracting the set in $\fset$ that contains $u_i$.
Therefore, instead of finding a collection $\pset$ of paths with the above properties, we will compute a collection $\qset$ of paths in $G$, such that
\begin{properties}{Q}
    \item each path $Q\in \qset$ connects a pair $s_Q,t_Q$ of terminals; \label{prop: Qterminal}
    \item the total $G$-distances between endpoints is large: $\sum_{Q\in \qset}\dist_G(s_Q,t_Q)=n\cdot (\log n)^{\Omega(1)}$; and \label{prop: Qdistance}
    \item their induced paths in $H$, which we denote by $\pset$, are almost edge-disjoint in $H$. \label{prop: edge-disjoint}
\end{properties}


The construction of the set $\qset$ of paths consists of two phases. In the first phase, we relax Property~\ref{prop: Qterminal} by allowing the endpoints of paths in $\qset$ to be any vertices in $G$, and in the second phase, we connect these endpoints to terminals, achieving Property~\ref{prop: Qterminal}. 
Below we describe two phases in more detail.

As a final preparation before we go into details, we introduce a parameter $\alpha$.
In Property \ref{prop: Qdistance}, suppose we want the total terminal-pair distances to be $\Omega(n\cdot (\log n)^{1-\alpha})$, where $0<\alpha<1$ is a constant to be determined later. 
This parameter $\alpha$ poses the following some requirements on the set $\qset$ of paths we are trying to find in Phase $1$.

\textbf{Requirement 1. the average length of paths in $\pset$ is $O\big((\log n)^{\alpha}\big)$.}
Since graph $G$ is a constant-degree expander, it has diameter $O(\log n)$, and so the distance between every pair $(s_Q,t_Q)$ of vertices in $G$ is at most $O(\log n)$. Therefore, to achieve total distance $\Omega(n\cdot (\log n)^{1-\alpha})$, we need $M=\Omega(n/(\log n)^{\alpha})$ paths.
Since Property \ref{prop: edge-disjoint} requires that paths in $\pset$ are (almost) edge disjoint in $H$, on average they may have length at most $O\big((\log n)^{\alpha}\big)$.

\textbf{Requirement 2. $k\ge n/2^{(\log n)^{\alpha}}$.} It is more a ``restriction'' than a requirement, that  vertices are not too many more than terminals in $G$. Its necessity comes from Phase 2 and will be explained later in this section. Another use of parameter $\alpha$ facilitated by this requirement is that, in the contracted graph $H$, we will think of $k\approx n/2^{(\log n)^{\alpha}}$ and allow $k\cdot 2^{(\log k)^{\alpha-o(1)}}$ Steiner nodes.

\subsubsection*{Phase 1. Finding the paths: hierarchical clustering}

We now describe the algorithm of finding a collection $\qset$ of paths that satisfies Properties~\ref{prop: Qdistance},\ref{prop: edge-disjoint}. 

We perform iterations, where in each iteration, we compute a $G$-path $Q\in \qset$ and its induced $H$-path $Q'\in \pset$, and remove the edges of $E(Q')$ from both $H$ and $G$ (since $H$ is a contracted graph of $G$, each edge of $H$ is also an edge of $G$), and the next iteration is continued on the remaining graphs $G$ and $H$.
Since expanders are robust against edge deletions \cite{saranurak2019expander}, as long as the number of iterations is at most $O\big(n/(\log n)^{\alpha}\big)$, we can think of each iteration as being performed on a fresh constant-degree expander and its contracted graph. For the purpose of technical overview, we will only describe how to perform the first iteration and find the first paths $Q$ and $Q'$. From the above discussion, the length of $Q$ needs to be at least $\Omega(\log n)$ and the length of $Q'$ needs to be at most $O\big((\log n)^{\alpha}\big)$.

To start with, we observe that a node in $H$ corresponds to a cluster in $G$, so any $G$-path lying entirely within such a cluster only induces an empty $H$-path.
Specifically, let $\fset$ be the partition of $V(G)$ into clusters. Since $k= n/2^{(\log n)^{\alpha}}$ and we allow $k\cdot 2^{(\log k)^{\alpha-o(1)}}$ Steiner nodes, 
$$|\fset|\le k\cdot 2^{(\log k)^{\alpha-o(1)}} =  n/\left(2^{(\log n)^{\alpha}-(\log n)^{\alpha-o(1)}}\right)\leq O\bigg(n/\left(2^{(\log n)^{\alpha}}\right)\bigg).$$
Therefore, some cluster $F\in \fset$ has size at least $|F|\ge \Omega(2^{(\log n)^{\alpha}})$, and as a subgraph of a constant-degree graph, $F$ must contain a path of length $\Omega\big((\log n)^{\alpha}\big)$, that consumes no edges from $H$.

As we would like paths in $\qset$ to have length $\Omega(\log n)$, such a path is not yet long enough.
As longer paths reside in larger clusters, the idea is to merge clusters in $\fset$ with their neighboring ones to obtain larger ones.
Specifically, we use a standard greedy algorithm to find a maximal subset $C$ of nodes in $H$ such that every pair of nodes in $C$ are at distance at least $3$ in $H$.\footnote{We start with an empty set $C$, and iteratively add to it an arbitrary node in $H$ that is at distance at least $3$ from all nodes in $C$, until we cannot do so.} Nodes in $C$ are called centers.
For each remaining node in $V(H)\setminus C$, either (i) it is at distance $1$ to a unique center, in which case we assign it to this center; or (ii) it is at distance $2$ to possibly a few centers, in which case we assign it to any such center. Now each center, together with all nodes assigned to it, form a larger cluster, that we call a \emph{level-$1$ cluster} (singleton vertices in $H$ are viewed as \emph{level-$0$ clusters}). And then sequentially for $i=2,3,\ldots$, we build level-$i$ clusters by merging level-$(i-1)$ clusters in the same way. 

To analyze the growth of cluster size, we make the following simplifying assumptions.
\begin{enumerate}
    \item \textbf{all level-$i$ clusters in $H$, when expanded as clusters in $G$, have the same size $s_i$.} In fact, we will enforce an approximate version of this property by discarding clusters with irregular size. We can show that we only discard a small number of vertices at each level, which can be accommodated by the robustness of expanders.
    \item \textbf{for each $i\ge 1$, $s_{i+1}=\Theta(s_i^2)$.} The intuition is that, as each level-$i$ center has size $s_i$, it has $\Theta(s_i)$ out-edges in $G$ as $G$ is a constant-degree expander. We expect most of the out-edges to land in distinct level-$i$ clusters, each of size $s_i$. Therefore, the level-$(i+1)$ cluster formed by this center will contain $\Theta(s_i)\cdot s_i=\Theta(s_i^2)$ vertices.
\end{enumerate}

\begin{figure}[h]
	\centering
	\scalebox{0.1}{\includegraphics{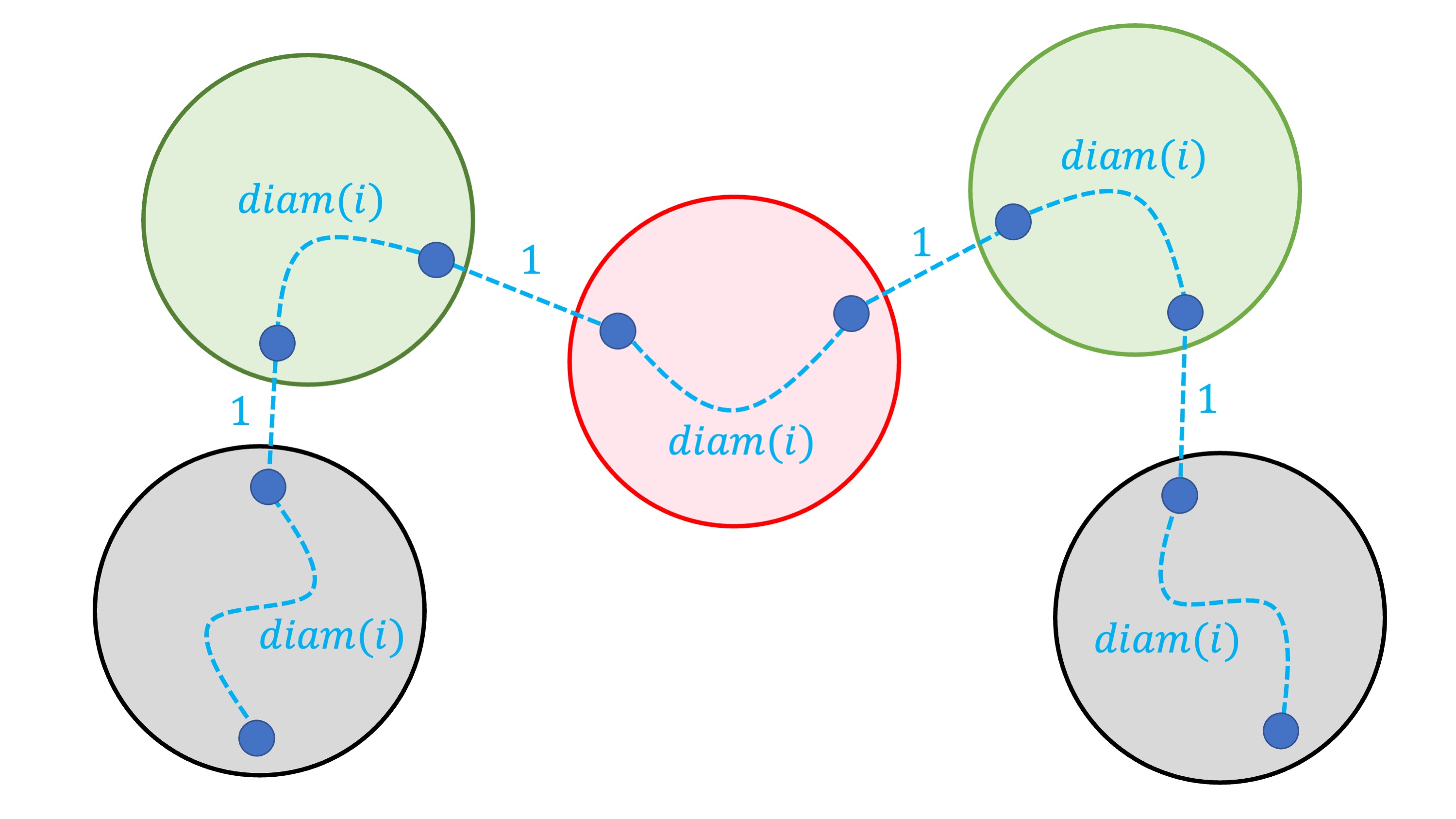}}
	\caption{An explanation of why the diameter increases by $5$ times per level up: A level-$(i+1)$ cluster consists of five level-$i$ clusters, with the red one being the center, the green ones its distance-$1$ neighbors, and the black ones distance-$2$ clusters assigned to it. The blue dashed line shows the shortest path between a pair of vertices in this level-$(i+1)$ cluster, whose length is at most $5\cdot \textnormal{diam}(i)+4$.\label{fig: diam}}
\end{figure}

From the merging algorithm and the explanation in \Cref{fig: diam}, in graph $H$, the maximum diameter of a level-$(i+1)$ cluster is at most $5$ times the maximum diameter of a level-$i$ cluster, indicating that the $H$-edge consumption is increased by $5$ times every level up. Therefore, the ultimate value of $\alpha$ is determined by the race of $G$-cluster size and $H$-edge consumption:
\begin{itemize}
    \item $G$-cluster size: starting at $2^{(\log n)^{\alpha}}$ (the size of the expanded $G$-cluster of a singleton node in $H$), being squared every level up, until it gets $n^{\Omega(1)}$ (so that it contains a $G$-path of length $\Omega(\log n)$).
    \item $H$-edge consumption: starting at $1$, being multiplied by $5$ every level up, until it gets $O\big((\log n)^{\alpha}\big)$ (by Requirement $1$).
\end{itemize}

In order for the Phase 1 to succeed, the $G$-cluster size must win this race: it must grow to $n^{\Omega(1)}$ before the $H$-edge consumption grows to $O\big((\log n)^{\alpha}\big)$. Therefore,
\[
\log_2 \bigg(\frac{\Omega(\log n)}{(\log n)^{\alpha}}\bigg)\le \log_5 \bigg((\log n)^{\alpha}\bigg).
\]
This requires that $\alpha \ge \frac{\log_2 5}{\log_2 5 + 1}$. As a consequence, the best quality lower bound we  can prove via this approach is $(\log n)^{1-\alpha}=(\log n)^{\frac{1}{\log_2 5+1}}$, as shown in \Cref{thm: main}.

\subsubsection*{Phase 2. Connecting the paths to terminals: surgery on edge capacities}

In the first phase, we found a collection $\qset$ of $M=\Omega(n/(\log n)^{\alpha})$ paths satisfying Properties~\ref{prop: Qdistance},\ref{prop: edge-disjoint}. However, their endpoints can be any vertices in $G$.
In this phase, we describe how to connect these endpoints to terminals via almost edge-disjoint paths, achieving Property~\ref{prop: Qterminal}. 

As we allow $k\cdot 2^{(\log k)^{\Omega(1)}}$ Steiner nodes, $k$ has to be bounded away from $n$ by at least factor $2^{(\log n)^{\Omega(1)}}$, 
which means that $k\ll M$. Since $G$ is a constant-degree graph, each terminal has only $O(1)$ incident edges, and so it is impossible to connect all $2M$ endpoints of $M$ paths to terminals by almost edge-disjoint paths. 

To overcome this issue, the idea is to increase the edge capacities in $G$. Imagine if each terminal has degree $M/k$ instead of $O(1)$ (while other vertices still have constant degree), then from the expansion property of $G$ and we can route the $2M$ endpoints to $k$ terminals with congestion $O(1)$. Since $G$ is a constant-degree expander, if we start from a terminal $t$ and perform BFS in $G$, then at the $(\log (M/k)$)-th level we can reach $M/k$ vertices. We will increase the capacities on all edges within these $\log (M/k)$ levels to make it look like that $t$ can reach each of these $M/k$ vertices at the last level with a distinct edge. Specifically, for each $i\le \log (M/K)$, we denote by $E_i$ the set of all edges between level-$i$ and level-$(i+1)$ vertices in this local BFS tree around $t$, and we give each edge in $E_i$ capacity $(M/k)/|E_i|$, so that their total capacity is $M/k$. See \Cref{fig: capacity} for an illustration.

\begin{figure}[h]
	\centering
	\subfigure[A terminal $t$ with its local BFS up to level $\log (M/k)$ giving a complete binary tree. Each edge on level $i$ has capacity $M/(2^i\cdot k)$.]
	{\scalebox{0.1}{\includegraphics{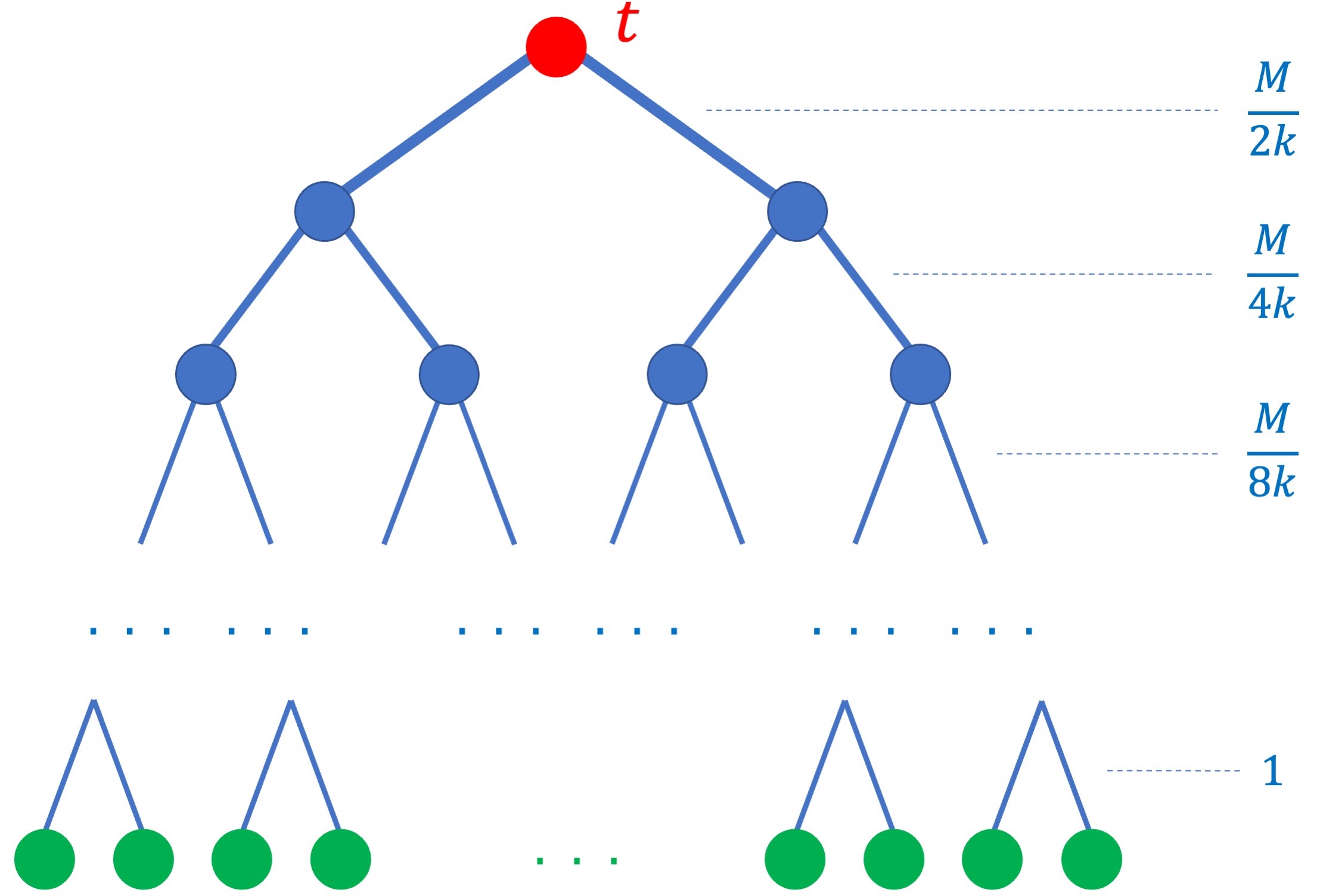}}}
	\hspace{0.5cm}
	\subfigure[A star graph connecting $t$ to each of its local BFS level $\log (M/k)$ children with a capacity-$1$ edge. The star is flow equivalent to the left graph.]
	{
		\scalebox{0.1}{\includegraphics{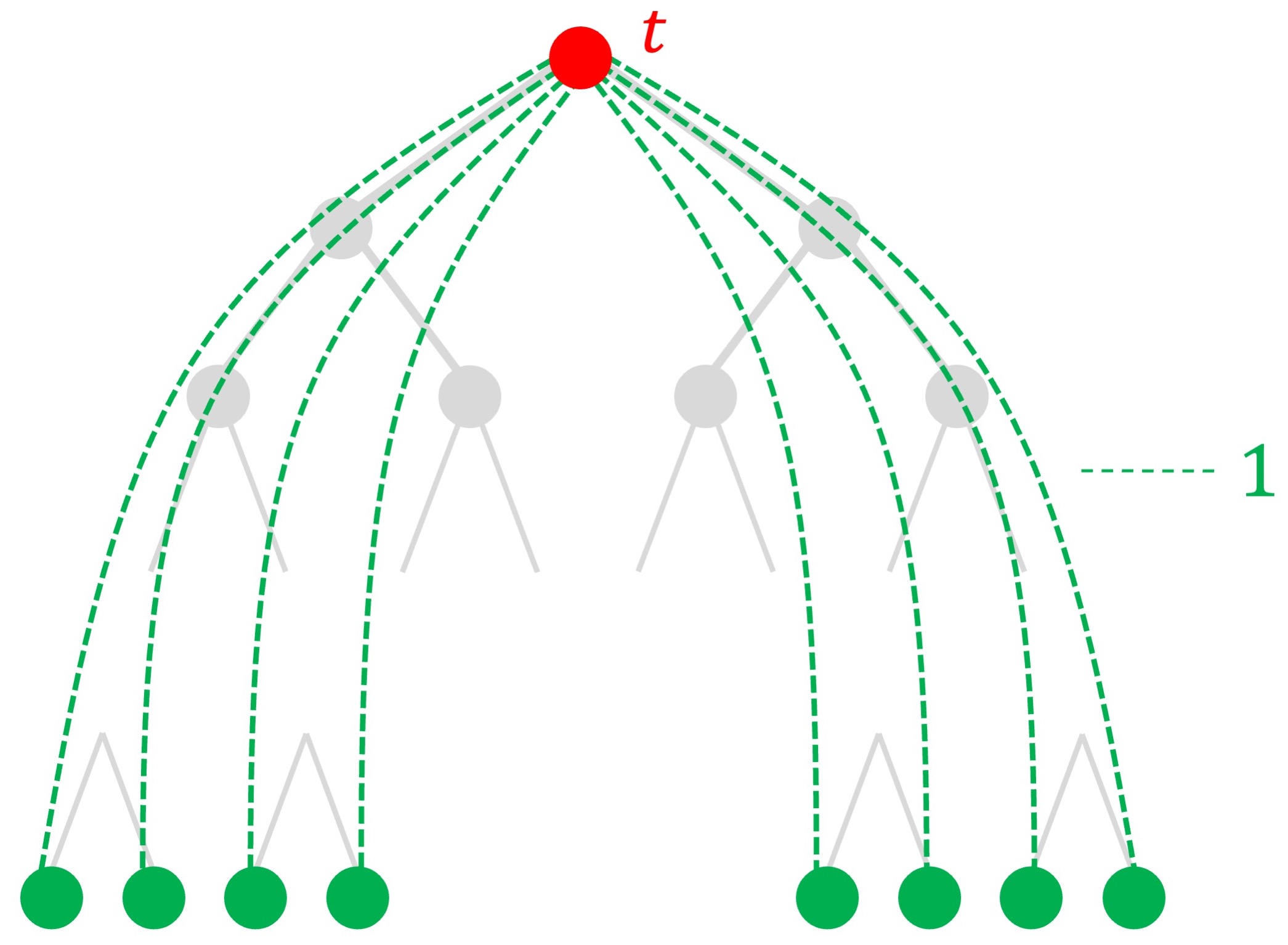}}}
	\caption{An illustration of increasing edge capacities.\label{fig: capacity}}
\end{figure}

In the new edge-capacitated graph, which we denote by $G'$, one can easily prove using its expansion property that all $2M$ endpoints can be connected to terminals with congestion $O(1)$.

\paragraph{Necessity of Requirement 2 ($k\ge n/2^{(\log n)^{\alpha}}$).}
The new graph $G'$ has the same graph structure as $G$ but higher edge capacities. It is therefore more powerful in routing demands. However, we still want to ensure that the constructed demand $\dset^*$ incur a large congestion in $G'$ using a similar argument as Inequality \ref{eqn: cong}. Therefore, instead of using $|E(G)|=O(n)$ in Inequality \ref{eqn: cong}, we need that the total capacities of all edges in $G'$ remains $O(n)$. As the increase in capacities are only in the local terminal BFS trees, and in each tree, the increase is at most $M/k$ in each of its $\log(M/k)$ levels, resulting in a total increase of $k\cdot \log(M/k)\cdot M/k$. Therefore, we will require that
\[k\cdot \log(M/k)\cdot M/k =O(n).\]
Plugging in $M=\Omega(n/(\log n)^{\alpha})$, we get that $k\ge n/2^{(\log n)^{\alpha}}$. 

\subsection{Related Work}
\label{sec: related}

\paragraph{Other recent results on flow sparsifiers.}
Chuzhoy \cite{chuzhoy2012vertex} constructed quality-$O(1)$ contraction-based flow sparsifiers of size $C^{O(\log\log C)}$, where $C$ is the total capacity of all terminal-incident edges (assuming that each edge has capacity at least $1$).
Andoni, Gupta, and Krauthgamer \cite{andoni2014towards} constructed quality-$(1+\eps)$ flow sparsifiers for quasi-bipartite graphs and quality-$1$ flow sparsifiers for Okamura-Seymour instances. 
On the other hand, Krauthgamer and Mosenzon \cite{krauthgamer2023exact} showed that even $6$-terminal graphs require arbitrarily large size in quality-$1$ flow sparsifiers, and quality-$(1+\eps)$ flow sparsifiers for $k$-terminal networks must contain at least $\Omega(k/\poly(\eps))$ Steiner nodes.
Later, Chen and Tan \cite{chen20241+} showed that, for contraction-based flow sparsifiers, even $6$-terminal graphs require arbitrarily large size for quality $(1+\eps)$.
Das, Kumar and Vaz \cite{das2024nearly} studied exact flow sparsifiers for quasi-bipartite graphs via a connection to polyhedral geometry.

\paragraph{Cut sparsifiers.}
A closely related notion is cut-preserving vertex sparsifiers (abbreviated as \emph{cut sparsifiers}).
Given a graph $G$ and a set $T$ of its $|T|=k$ vertices called terminals, a cut sparsifier of $G$ with respect to $T$ is a graph $G'$ that contains $T$, and for every partition $(T_1,T_2)$ of $T$, the min-cut separating $T_1$ from $T_2$ in $G$ is within a multiplicative factor $q$ from the min-cut separating $T_1$ from $T_2$ in $G'$, where $q$ is called its \emph{quality}. Quality-$q$ flow sparsifiers are also quality-q cut sparsifiers, but the converse is not necessarily true. In fact, by using the previous results on flow-cut gap, we can show that quality-$q$ cut sparsifiers are quality-$O(q\log k)$ flow sparsifiers.

In the restricted case where $V(G')=T$, the best quality is shown to be $O(\log k/\log\log k)$ \cite{moitra2009approximation} and $\Omega(\sqrt{\log k}/\log\log k)$ \cite{makarychev2016metric,charikar2010vertex}. 
In another restricted case where $q=1$, the size bound is shown to be at most $2^{2^k}$ \cite{hagerup1998characterizing,khan2014mimicking} and at least $2^{\Omega(k)}$ \cite{khan2014mimicking,krauthgamer2013mimicking,karpov2019exponential}.
In the special setting where each terminal has degree $1$ in $G$, Chuzhoy \cite{chuzhoy2012vertex} constructed $O(1)$-quality cut sparsifiers of size $O(k^3)$, and Kratsch and Wahlstrom \cite{kratsch2012representative} constructed quality-$1$ cut sparsifiers of size $O(k^3)$ via a matroid-based approach.

Recent years have witness other work on (i) cut/flow sparsifiers for special families of graphs, such as trees \cite{goranci2017vertex}, planar graphs \cite{krauthgamer2020refined,karpov2019exponential,goranci2020improved,chen2025cut}, and quasi-bipartite graphs \cite{das2024nearly,chen2025cut}; (ii) preserving terminal min-cuts up to some threshold via contraction-based cut sparsifiers \cite{chalermsook2021vertex,liu2023vertex}; and
(iii) constructions of dynamic cut/flow sparsifiers and their applications in dynamic graph algorithms \cite{abraham2016fully,durfee2019fully,chen2020fast,goranci2021expander}.

\section{Preliminaries}
\label{sec: prelim}

By default, all logarithms are to the base of $2$.

Let $G = (V, E, c)$ be a undirected graph with node set $V=V(G)$ and edge set $E=E(G)$, where each edge $e \in E$ has capacity $c(e)\geq 1$. Let $\dist_G(v,v')$ be the shortest-path distance between $v$ and $v'$ in $G$. 
For a vertex $v \in V$, we denote by $\deg_G(v)$ the degree of $v$ in $G$. 
For a subset $E'\subseteq E$, let $|E'|$ be cardinality of $E'$ and define $c(E')=\sum_{e\in E'}c(e)$. Set $c(\emptyset)=0$ by default. 

For any disjoint subsets $S,T \subseteq V$, we denote by $E_G(S,T)$ the set of edges in $G$ with one endpoint in $S$ the other endpoint in $T$. For any subset $U\subseteq V$, we denote by $\delta_G(U)=E_G(U,V\setminus U)$ the set of edges in $G$ with one endpoint in $U$ and the other endpoint in $V\setminus U$, and denote by $N_G(U)$ the set of endpoints of $\delta_G(U)$ outside of $U$.

We may omit the subscript $G$ in the above notations when the graph is clear from the context.

\paragraph{Partitions and contracted graphs.}
Given graph $G=(V,E,c)$ and any partition $\mathcal{F}$ of $V$, for each node $u\in V$, we denote by $F(u)\subseteq V$ the cluster in $\mathcal{F}$ containing node $u$. Given any subset $\mathcal{F}'\subseteq \mathcal{F}$, set $V(\mathcal{F}')=\bigcup_{F\in \mathcal{F}'}F\subseteq V$. We denote by $H=(V_H,E_H,c_H)$ the graph obtained from $G$ by contracting each set $F\in \fset$ into a supernode (so $V_H=\mathcal{F}$), keeping all parallel edges and removing self-loops. For any distinct clusters $F,F'\in \mathcal{F}$, we say $(F,F')\in E_H$ if $E_G(F,F')\neq \emptyset$, and in this case we define $c_H(F,F')=c(E_G(F,F'))\geq 1$. We denote by $\dist_H(F,F')$ the shortest-path distance between $F$ and $F'$ in $H$. Given a pair $u,u'$ of nodes in $V$, we define $\dist_H(u,u')=\dist_H(F(u),F(u'))$.

\paragraph{Demands, flows, and congestions.} 
Given a graph $G=(V,E,c)$, let $T\subseteq V$ be a subset of vertices called \emph{terminals}. For each pair $t,t'$ of terminals, we denote by $\mathcal{P}^G_{t,t'}$ the collection of $t$-$t'$ paths in $G$, and we define $\mathcal{P}^G=\bigcup_{t,t'\in T}\mathcal{P}^G_{t,t'}$.

A multi-commodity flow $F=\{f_P\}_{p\in \mathcal{P}^G}$ on $G$ is a function that assigns to each path $P\in \mathcal{P}^G$ a real number $f_P\geq 0$. A demand $\mathcal{D}=\{\mathcal{D}_{t,t'}\}_{t,t'\in T}$ on $T$ is a function that assigns to each (unordered) pair $t,t'\in T$ a real number $\mathcal{D}_{t,t'}\geq 0$. Note that the demand value $\mathcal{D}_{t,t'}$ cannot be all zero.

We say that a flow $F$ \textit{routes} $\mathcal{D}$, if for every terminal pair $t,t'\in T$, $F$ sends $\mathcal{D}_{t,t'}$ units of flow from $t$ to $t'$ (or from $t'$ to $t$) in $G$, e.g. $\sum_{P\in \mathcal{P}^G_{t,t'}}f_P = \mathcal{D}_{t,t'}$. 
Consider a multi-commodity flow $F$, the \emph{congestion} of flow $F$ in $G$, denoted by $\congestion_G(F)$, is the maximum, over all edges $e$ in $G$, the ratio between the amount of flow sent through $e$ and its capacity $c(e)$: 
\[
\congestion_G(F)=\max_{e\in E} \bigg\{\frac{\sum_{P\in \mathcal{P}^G:e\in P}f_P}{c(e)}\bigg\}.
\]
Given a graph $G$ and a demand $\mathcal{D}$, the congestion of $\mathcal{D}$ in $G$, denoted by $\text{cong}_G(\mathcal{D})$, is defined to be the minimum congestion of any flow that routes $\mathcal{D}$ in $G$:
$\congestion_G(\mathcal{D})=\min \{\congestion_G(F):F \text{ routes }\mathcal{D}\}$.
\begin{lemma}\label{lem:cong>=dist/c(E)}
For any graph $G=(V,E,c)$, any terminal set $T\subseteq V$ and any demand $\mathcal{D}$ among terminals,
    \[ \congestion_G(\mathcal{D})\geq \frac{\sum_{t,t'\in T} \mathcal{D}_{t,t'}\cdot \dist_G(t,t')}{c(E)}.
    \]
\end{lemma}
\begin{proof}
By the definition of the congestion of demand $\mathcal{D}$, there exists a flow $F^*=\{f_P^*\}_{P\in \mathcal{P}}$ on $G$ such that $\congestion_G(\mathcal{D})=\congestion_G(F^*)$ and $F^*$ routes $\mathcal{D}$. By the definition of the congestion of flow $F^*$, for each edge $e\in E$, it has $\congestion_G(F^*)\cdot c(e) \geq \sum_{P\in \mathcal{P}^G: e\in P}f^*_P$. By summing over all edges in $G$, we have
\[
\begin{split}
    \congestion_G(F^*)\cdot c(E) &=
\sum_{e\in E}\congestion_G(F^*)\cdot c(e) 
\geq \sum_{e\in E}\sum_{P\in \mathcal{P}^G:e\in P}f^*_P
=\sum_{P\in \mathcal{P}^G}f^*_P\cdot |E(P)|
=\sum_{t,t'\in T}\sum_{P\in \mathcal{P}_{t,t'}^G}f^*_P\cdot |E(P)|\\
& \geq \sum_{t,t'\in T}\sum_{P\in \mathcal{P}_{t,t'}^G}f^*_P\cdot\dist_G(t,t') 
= \sum_{t,t'\in T}
\mathcal{D}_{t,t'}
\cdot\dist_G(t,t'),
\end{split}
\]
since $\dist_G(t,t')=\min_{P\in \mathcal{P}^G_{t,t'}}|E(P)|$ denotes shortest-path distance between $t,t'$ in $G$.
\end{proof}
\paragraph{Flow sparsifiers and their quality.}
Let $H$ be a graph with $T\subseteq V(H)$. We say that $H$ is a flow sparsifier of $G$ with respective to $T$ with quality $q\geq 1$, if for any demand $\mathcal{D}$ on $T$,
\begin{equation}\label{eq:cong}
    \text{cong}_H(\mathcal{D})\leq \text{cong}_G(\mathcal{D})\leq q\cdot\text{cong}_H(\mathcal{D}).   
\end{equation}
We say that a partition $\mathcal{F}$ of $V$ is valid if terminals in $T$ lie in distinct clusters in $\mathcal{F}$.
A graph $H$ is a \emph{contraction-based flow sparsifier} of $G$ with respective to $T$, if there exists a valid partition $\mathcal{F}$ of vertices in $G$ into clusters, such that $H$ is obtained from $G$ by contracting vertices in each cluster $F\in \mathcal{F}$ into a supernode, keeping parallel edges and discarding self-loops. The \emph{size} of $H$ is defined as the number of its vertices, namely $|\fset|$.
A \emph{convex combination of contraction-based flow sparsifier} is a distribution $\mu$ on contraction-based flow sparsifiers. We say that the convex combination $\mu$ has size $s$ iff every contraction-based flow sparsifier $H$ in its support (with non-zero probability) has size at most $s$. 
A underlying graph $H_{\mu}$ is defined as follows. We take a copy of every graph $H$ in the support of $\mu$, scale its edge capacity by $\Pr_\mu[H]$, and then glue them together at all terminals (since each such graph must contain all terminals).
The quality of the convex combination $\mu$ is simply defined as the quality of $H_{\mu}$.

\section{Proof of \Cref{thm: main} for Contraction-based Flow sparsifiers}\label{sec:pf main thm}
In this section, we provide the proof of \Cref{thm: main} for contraction-based flow sparsifiers, with some technical details deferred to \Cref{sec:pf algo-property}. 
We generalize the proof to convex combinations of contraction-based flow sparsifiers in \Cref{sec: convex combination}.
The heart of this section is following theorem.

\begin{restatable}{thm}{finalresult}\label{thm:final result}
Let $\varepsilon >0$ be a small constant, $\alpha=\frac{\log 5}{\log 5+1-\varepsilon}$, and $\beta = \frac{\log 5-0.5\varepsilon}{\log 5+1-\varepsilon}$. Then for each large enough integer $n$, there exists an $n$-vertex graph $G=(V,E,c)$ with a terminal set $T$ of size $\frac{n}{2^{(\log n)^{\alpha}}}$, such that for any valid partition $\mathcal{F}$ of $V$ with size $|\mathcal{F}|=O(\frac{n}{2^{(\log n)^\beta}})$, there exists a demand $\mathcal{D}$ on terminals with $\congestion_G(\mathcal{D})=\Omega\left((\log n)^{1-\alpha}\right)$ and $\congestion_H(\mathcal{D})=O(1)$ where $H$ is the contracted graph of $G$ with respect to $\fset$.

\end{restatable}
We now show that \Cref{thm:final result} implies \Cref{thm: main}.
Let $\eps',k$ be the parameters given in \Cref{thm: main}. We define $\eps=20\eps'$, $\alpha=\frac{\log 5}{\log 5+1-\varepsilon}$, $\beta = \frac{\log 5-0.5\varepsilon}{\log 5+1-\varepsilon}$, and
let $n=\lfloor k\cdot 2^{{(\log k)}^{\alpha+2\eps'}}\rfloor$.
Observe that $\log n\geq \log k =\log n - (\log n)^{\alpha}\geq 0.99\log n$.
%
Consider the graph instance $(G,T)$ given by \Cref{thm:final result}.
Let $H$ be the contraction-based flow sparsifier of $G$ determined by the valid partition $\mathcal{F}$
with 
$$|V(H)|=|\mathcal{F}|\leq k\cdot 2^{(\log k)^{\alpha+\varepsilon'}}
=n/\left(2^{(\log k)^{\alpha+2\eps'}-(\log k)^{\alpha+\eps'}}\right)
\leq n/2^{(\log k)^{\beta}}
=O\bigg(\frac{n}{2^{(\log n)^\beta}}\bigg).$$
From \Cref{thm:final result}, there is a demand $\dset$ such that $\congestion_G(\mathcal{D}) = \Omega\left((\log n)^{1-\alpha}\right)$ and 
$\congestion_H(\mathcal{D}) = O(1)$. 
Therefore, the quality of $H$ is
$q\geq \frac{\congestion_G(\mathcal{D})}{\congestion_H(\mathcal{D})}
= \Omega\left((\log n)^{1-\alpha}\right)
= (\log k)^{\frac{1}{\log 5 + 1}-O(\eps)}$.

\textbf{Organization of this section.} The remainder of this section is dedicated to the proof of \Cref{thm:final result}. 
In \Cref{subsec:rg-property} and \Cref{subsec:edge cap}, we construct an $n$-vertex graph $G=(V,E,c)$ and the terminal set $T$ used in \Cref{thm:final result}.
Then in \Cref{subsec:demand}, we propose a strategy to construct the demand over terminals, given any valid partition. 
Finally, in \Cref{subsec:completing pf final result}, we complete the proof of \Cref{thm:final result}.
The proofs of certain technical lemmas are deferred to \Cref{sec:pf edge disjoint path}.

\subsection{Expanders from random matchings}\label{subsec:rg-property}


Let $V$ be a set of $|V|=n$ vertices for an even integer $n$. We construct a random graph $G$ on $V$ by sampling uniformly at random $10$ perfect matchings on $V$ and setting $G$ as the union of them, so $G$ is a regular graph with degree $d=10$. 
This graph $G$ is at the heart of the proof of \Cref{thm:final result} and will be used throughout this section. Here we first state and prove some of its properties.
Recall that the conductance of a graph $H$ is defined as
    \[
    \Phi(H)=\min_{S\subseteq V(H), S\neq\emptyset, S\neq V(H)} \left\{\frac{|E_H(S)|}{\min \{\sum_{v\in S}\deg_H(v), \sum_{v\notin S}\deg_H(v)\}}\right\}.
    \]
We say that $H$ is a $\phi$-expander if $\Phi(H) \geq \phi$. 

\begin{lemma}[\cite{puder2015expansion}]\label{lemma:G'-expander}
    With probability $1-o(1)$, $\Phi(G)\geq \frac{1}{2}$.
\end{lemma}

\begin{lemma}\label{lem:neighbor-size>=2}
    With probability $1-o(1)$, for every subset $U\subseteq V$ of size $|U|\leq 0.5n$, $|N(U)|\geq 0.5|U|$.
\end{lemma}
\begin{proof}
Assume $\Phi(G)\geq \frac{1}{2}$ (which holds with probability $1-o(1)$ according to \Cref{lemma:G'-expander}). Therefore,
for any subset $U\subseteq V$ of size $|U|\leq 0.5n$, $|\delta_G(U)|\geq 0.5 d\cdot|U|$. 
Since $\delta_G(U)=E_G(U,N(U))$, $|N(U)|\geq |E_G(U,N)|/d\geq 0.5|U|$. 
\end{proof}

\subsection{Construction of the hard instance: graph and edge capacities}\label{subsec:edge cap}\label{subsec:contruct G}

We continue to use the random graph $G$ defined in \Cref{subsec:rg-property}.
Recall the parameters $\alpha=\frac{\log 5}{\log 5+1-\varepsilon}$ and $\beta=\frac{\log 5-0.5\varepsilon}{\log 5+1-\varepsilon}$ from \Cref{thm:final result}.
We define the terminal set $T$ of $G$ to be an arbitrary set $|T|=k=\lfloor\frac{n}{2^{(\log n)^\alpha}}\rfloor$ vertices in $G$. For each integer $i\geq 0$, construct the following sets:
\begin{itemize}
    \item $N_i=\{u\in V: \min_{t\in T} \{\dist_G(u,t)\}= i\}$ 
    \item $B_i=\{u\in V: \min_{t\in T} \{\dist_G(u,t)\}\leq i\}=\bigcup_{j=0}^i N_j$.
\end{itemize}
We define $m=\lfloor\frac{10n}{(\log n)^{\alpha}}\rfloor$ and $r=\min\{i: |B_i|\geq 2m\}-1$. By definition of $r$, $|B_{r+1}|\geq 2m$. Since $G$ is a degree-$10$ regular graph, $0.2m\leq |B_{r}| < 2m$.

\begin{lemma}\label{lem:ball-lb}
    With probability $1-o(1)$, $r=O((\log n)^{\alpha})$.
\end{lemma}
\begin{proof}
    By the definition of $r$, $|B_{i}|\leq 0.5n$ for each $0\le i\le r$. From \Cref{lem:neighbor-size>=2}, with probability $1-o(1)$, for each $i$, $|B_{i+1}|\geq (1.5)\cdot|B_{i}|\geq (1.5)^{i+1}\cdot |T|$.
    Since $(1.5)^r\cdot|T| \leq |B_r|< 2m$, $(1.5)^r < \frac{2^{(\log n)^\alpha}}{n}\cdot\frac{20n}{(\log n)^{\alpha}}=\frac{20\cdot2^{(\log n)^\alpha}}{(\log n)^{\alpha}}$ and therefore $r=O((\log n)^{\alpha})$.
\end{proof}

\paragraph{Edge capacities.} 
Let $V_T=B_{r}$ and $E_{r+1} = E_G(N_r,N_{r+1})$. For each node $u\in N_r$, let $E_u$ be the set of all incident edges of $u$ in $E_{r+1}$, so $E_{r+1}=\bigsqcup_{u\in N_r} E_u$. We construct a graph $G_T=(V_T, E_T, c')$ as follows. 
\begin{enumerate}
    \item Initially $E_T\leftarrow \emptyset$. 
    \item For each node $u\in N_r$, 
        \begin{enumerate}
            \item Pick an arbitrary edge $e_u=(u,u')\in E$ such that $u'\in N_{r-1}$.
            \item Set $E_T \leftarrow E_T \cup \{e_u\}$ and $c'(e_u)=|E_u|\geq 0$.
        \end{enumerate}
    \item Go through each integer $i=r-1,\ldots,2,1$ sequentially, for each node $u\in N_i$, 
    \begin{enumerate}
        \item Set up $E_u=\{(u,u'')\in E_T\}$.
        \item Pick an arbitrary edge $e_u=(u,u')\in E$ such that $u'\in N_{i-1}$.
        \item Set $E_T \leftarrow E_T \cup \{e_u\}$ and $c'(e_u)=c'(E_u) = \sum_{e\in E_u}c'(e)$.
    \end{enumerate}
\end{enumerate}
Note that it is possible that some node $u$ has $E_u=\emptyset$ and $c'(e_u)=0$.
For each edge $e\in E$, we set its edge capacity as follows:
\begin{equation}\label{eq:edge cap}
\text{for each edge } e\in E, \quad c(e)=
\begin{cases}
    \min \{c'(e),1\}\quad & \text{if } e\in E_T \\
    1 \quad &  \text{otherwise}
\end{cases}
\end{equation}
The hard instance is $G=(V,E,c)$, the graph $G$ with edge capacity $\{c(e)\}_{e\in E}$ defined in (\ref{eq:edge cap}). We prove the following properties.

\begin{restatable}{lem}{edgedisjointpath}\label{thm:edge-disjoint-path}
For any node set $U\subseteq V$ of size $m=\frac{10n}{(\log n)^\alpha}$, there exists a collection $\mathcal{Q}^G_U$ of $m$ edge-disjoint paths in $G$, each connecting a distinct vertex of $U$ to a terminal in $T$.    
\end{restatable}

The proof of \Cref{thm:edge-disjoint-path} is deferred to \Cref{sec:pf edge disjoint path}.

\begin{lemma}\label{lem: c(E)=O(n)}
    With probability $1-o(1)$, $G=(V,E,c)$ has total edge capacity $c(E)=\sum_{e\in E}c(e)=O(n)$.
\end{lemma}
\begin{proof}
    Since $c(E)=\sum_{e\in E} c(e)
    \leq \sum_{e\in E_T} c(e) + |E|
    \leq c'(E_T)+|E|$ and $|E|=O(n)$, it suffices to prove that $c'(E_T)=\sum_{e\in E_T}c'(e)=O(n)$.
For each $1\le i\le r$, define $E_i=\bigcup_{u\in N_i}\{e_u\}$. 
    Note that $E_T=\bigsqcup_{i=1}^{r} E_i$.
    For each $i$, $E_i$ can be partitioned as
    $E_i=\bigsqcup_{u\in N_{i-1}}E_u$ and for each $1\le i\le r-1$,
    \[
    c'(E_{i})=c'\left(\bigcup_{u\in N_i} \{e_u\}\right) = \sum_{u\in N_i}c'(e_u) = \sum_{u\in N_i}c'(E_u) = c'\left(\bigsqcup_{u\in N_i} E_u\right)=c'(E_{i+1}).
    \]
    Since $E_{r+1}=E_G(N_r,N_{r+1})=\bigsqcup_{u\in N_r} E_u$,
    $c'(E_r)=\sum_{u\in N_{r}}c'(e_u)=\sum_{u\in N_{r}}|E_u|=|E_{r+1}|$. Additionally, $|E_{r+1}|\leq d\cdot |N_r|\leq d\cdot|B_r|=O(m)$ and $r=O((\log n)^\alpha)$ by \Cref{lem:ball-lb}. Therefore, 
    \[
    c'(E_T)=\sum_{e\in E_T}c'(e)= \sum_{i=1}^{r}c'(E_i)=r\cdot |E_{r+1}|\leq O\left((\log n)^\alpha\cdot \frac{n}{(\log n)^\alpha}\right)= O(n).
    \]
\end{proof}

\subsection{Construction of the hard instance: the demand on terminals}\label{subsec:demand}
Consider the edge-capacitated graph $G=(V,E,c)$ defined above. We assume that it satisfies all the properties in \Cref{subsec:rg-property} and \Cref{subsec:contruct G}. 
The next lemma states that the instance $(G,T)$ contains a collection of paths with desired properties, as discussed as the ``path-finding'' phase in the technical overview.
Its proof is provided in \Cref{sec:pf algo-property}.

\begin{restatable}{lem}{algoproperty}
\label{thm:algo-property}
Let $\varepsilon >0$ be a small constant, $\alpha=\frac{\log 5}{\log 5+1-\varepsilon}$, and $\beta = \frac{\log 5-0.5\varepsilon}{\log 5+1-\varepsilon}$.
Then with probability $1-o(1)$, for any valid partition $\mathcal{F}$ of $V$ with $|\mathcal{F}|=O(\frac{n}{2^{(\log n)^\beta}})$, 
    there exist a collection $\mathcal{P}$ of $m=\frac{10n}{(\log n)^{\alpha}}$ pairs of nodes in $V$ and a collection $\qset$ of $m$ edge-disjoint paths in the contracted graph $H$ (obtained from $G$ by contracting each set $F$ into a supernode), such that
    \begin{itemize}
        \item all endpoints in $\mathcal{P}$ are distinct;
        \item for each pair $(u,v)\in \mathcal{P}$, set $\qset$ contains 
        a path $Q_H(u,v)$ in $H$ connecting $F(u)$ to $F(v)$ with length at most $(\log n)^{\alpha}/2000$, where $F(w)$ denotes the cluster in $\fset$ containing $w$ for each node $w\in V$; and
        \item $\sum_{(u,v)\in \mathcal{P}}\dist_G(u,v) \geq \frac{1}{128}n(\log n)^{1-\alpha}$.
    \end{itemize}
\end{restatable}

\paragraph{Collection of terminal pairs.}
Consider any valid partition $\mathcal{F}$ of $V$ with size $|\mathcal{F}|=O(\frac{n}{2^{(\log n)^\beta}})$. Recall that we say a partition $\mathcal{F}$ of $V$ is valid if distinct terminals of $T$ belong to different sets in $\mathcal{F}$. Let $H$ be the graph obtained from $G$ by contracting each set of $\mathcal{F}$ into a supernode. 

Let $\mathcal{P}=\{(a_i,b_i)\}_{i=1}^{m}$ the set of $m=\frac{10n}{(\log n)^{\alpha}}$ node pairs and let $\mathcal{Q}=\{Q_H(a_i,b_i)\}_{i=1}^{m}$ be the collection of $m$ edge-disjoint paths in $H$ given by \Cref{thm:algo-property}. Define $A=\{a_1,a_2,\ldots, a_m\}$ and $B=\{b_1, b_2,\ldots, b_m\}$. From \Cref{thm:algo-property}, $A$ and $B$ are disjoint. We now construct a collection $\mathcal{P}_T=\{(t_i,t'_i)\}^m_{i=1}$ of $m$ terminal pairs and a collection $\mathcal{Q}_T=\{Q_H(t_i,t'_i)\}^m_{i=1}$ of $m$ corresponding paths in the following way:
\begin{enumerate}
    \item 
    By \Cref{thm:edge-disjoint-path}, there exists a collection $\mathcal{Q}^G_A=\{Q_G(a_i, t_i)\}_{i=1}^m$ of $m$ edge-disjoint paths such that for each $a_i\in A$, path $Q_G(a_i, t_i)$ connects $a_i$ to some terminal $t_i\in T$. 
    Let $Q_H(a_i,t_i)$ be its induced path in $H$, and denote $\mathcal{Q}^H_A=\{Q_H(a_i, t_i)\}_{i=1}^m$. Clearly, paths in $\mathcal{Q}^H_A$ are edge-disjoint paths.
    \item  By \Cref{thm:edge-disjoint-path}, there exists a collection $\mathcal{Q}^G_B=\{Q_G(b_i, t_i)\}_{i=1}^m$ of $m$ edge-disjoint paths such that for each $b_i\in B$, path $Q_G(b_i, t_i)$ connects $b_i$ to some terminal $t_i\in T$. 
    Let $Q_H(b_i,t_i)$ be its induced path in $H$, and denote $\mathcal{Q}^H_B=\{Q_H(b_i, t_i)\}_{i=1}^m$. Clearly, paths in $\mathcal{Q}^H_B$ are edge-disjoint paths.
    \item Go through each $i\in \{1,2,\ldots, m\}$, construct a new path $Q_H(t_i,t'_i)$ between $t_i$ and $t'_i$ in graph $H$ by concatenating paths $Q_H(t_i,a_i)$, $Q_H(a_i,b_i)$, and $Q_H(b_i,t'_i)$. 
    \item Set $\mathcal{P}_T=\{(t_i,t_i')\}_{i=1}^m$ and $\mathcal{Q}_T=\{Q_H(t_i,t_i')\}_{i=1}^m$.
\end{enumerate}
\paragraph{Demand over terminals.} Given any valid partition $\mathcal{F}$ of $V$ with size $|\mathcal{F}|=O(\frac{n}{2^{(\log n)^\beta}})$, set up the collection of terminal pairs $\mathcal{P}_T$. For each terminal pair $t,t'\in T$, set $\mathcal{D}_{t,t'}$ as the number of pair $(t,t')\in \mathcal{P}_T$ in the collection.
\subsection{Completing the proof of \Cref{thm:final result}}\label{subsec:completing pf final result}
In this section, we complete the proof of \Cref{thm:final result} by \Cref{lem:dist-demand-lb} and \Cref{lem:totalcost-lb}.

\begin{lemma}\label{lem:dist-demand-lb}
For any valid partition $\mathcal{F}$ of $V$ with  $|\mathcal{F}|=O(\frac{n}{2^{(\log n)^\beta}})$, the demand $\mathcal{D}$ constructed by \Cref{subsec:demand} satisfies that $\congestion_G(\mathcal{D})= \Omega\left((\log n)^{1-\alpha}\right)$.
\end{lemma}
\begin{proof}
By \Cref{lem:cong>=dist/c(E)}, we have $\congestion_G(\mathcal{D})\geq \sum_{t,t'\in T} \mathcal{D}_{t,t'}\cdot \dist_G(t,t')/c(E)$.
By the definition of demand, $\sum_{t,t'\in T} \mathcal{D}_{t,t'}\cdot \dist_G(t,t') = \sum_{(t_i,t_i')\in \mathcal{P}_T} \dist_G(t_i,t_i')$. In the following, we first prove that 
\[\sum_{(t_i,t_i')\in \mathcal{P}_T} \dist_G(t_i,t_i') = \Omega(n(\log n)^{1-\alpha}).
\]

For each $1\le i\le m$, by the construction in \Cref{subsec:demand}, there exist a path $Q_G(a_i,t_i)\in \mathcal{Q}_A^G$ and a path $Q_G(b_i,t'_i)\in \mathcal{Q}_B^G$. By triangle inequality, 
    \[\begin{split}
    \dist_G(a_i,b_i)
    & \leq \dist_G(a_i,t_i) + \dist_G(t_i,t'_i) + \dist_G(t'_i,b_i)\\
    & \leq 
    |Q_G(a_i,t_i)| + \dist_G(t_i,t'_i) + |Q_G(b_i,t'_i)|.
    \end{split}
    \]
    Summing over all pairs $(a_i,b_i)\in \mathcal{P}$, we get that
    \[
    \sum_{(a_i,b_i)\in \mathcal{P}}\dist_G(a_i,b_i)
    \leq \sum_{i=1}^m|Q_G(a_i,t_i)| + \sum_{(t_i,t_i')\in \mathcal{P}_T}\dist_G(t_i,t'_i) + \sum_{i=1}^m|Q_G(b_i,t'_i)|    \]
    Since paths in $\mathcal{Q}^G_A$ are edge-disjoint in $G$, $\sum_{i=1}^m |Q_G(a_i,t_i)| \leq |E|=O(n)$. Similarly, $\sum_{i=1}^m |Q_G(b_i,t'_i)|=O(n)$). From \Cref{thm:algo-property}, $\sum_{(a_i,b_i)\in \mathcal{P}}\dist_G(a_i,b_i)=
    \Omega\left(n(\log n)^{1-\alpha}\right)$. Therefore,
    \[
    \begin{split}
    \sum_{(t_i,t'_i)\in \mathcal{P}_T}\dist_G(t_i,t'_i)
    & \geq 
    \sum_{(a_i,b_i)\in \mathcal{P}}\dist_G(a_i,b_i)-
    \sum_{(a_i,t_i)\in \mathcal{P}_A}|Q_G^A(a_i,t_i)| -\sum_{(b_i,t'_i)\in \mathcal{P}_B}|Q_G^B(b_i,t'_i)|\\
    & \geq \Omega\left(n(\log n)^{1-\alpha}\right) - O(n) - O(n) = \Omega\left(n(\log n)^{1-\alpha}\right).
    \end{split}
    \]   
    From \Cref{lem: c(E)=O(n)}, $c(E)=O(n)$. Therefore,
    \[
    \congestion_G(\mathcal{D})\geq 
    \frac{\sum_{(t,t')\in \mathcal{P}_T} \dist_G(t,t')}{c(E)}\geq \frac{\Omega(n(\log n)^{1-\alpha})}{O(n)}=\Omega\left((\log n)^{1-\alpha}\right).
    \]    
\end{proof}

\begin{lemma}\label{lem:totalcost-lb}
For any valid partition $\mathcal{F}$ of $V$ with  $|\mathcal{F}|=O(\frac{n}{2^{(\log n)^\beta}})$, the demand $\mathcal{D}$ constructed by \Cref{subsec:demand} satisfies that 
$\congestion_H(\mathcal{D})\leq 3$, where $H$ is the contracted graph of $G$ with respect to $\fset$.
\end{lemma}

\begin{proof}
Consider graph $H$ contracted by partition $\mathcal{F}$. 
Let $\mathcal{P}_T=\{(t_i,t_i')\}_{i=1}^m$ and $\mathcal{Q}_T=\{Q_H(t_i,t_i')\}_{i=1}^m$ be the construction in \Cref{subsec:demand}. We now construct a multicommodity flow that routes $\dset$ in graph $H$ as follows: for each terminal pair $(t_i,t_i')\in \mathcal{P}_T$, send $1$ unit of flow along path $Q_H(t_i,t_i')$. 
Therefore $\congestion_H(\mathcal{D})\leq \congestion_H(F')$. In the following, we prove that $\congestion_H(F')\leq 3$.

Since each edge $e\in E(H)$ has capacity $c(e)\geq 1$, by the definition of congestion of flow, 
\[
\congestion_H(F') = \max_{e\in E(H)}\frac{\sum_{P\in \mathcal{P}^H:e\in P}f'_P}{c(e)}\leq \max_{e\in E} \sum_{P\in \mathcal{Q}_T:e\in P} 1.
\]
By the construction in \Cref{subsec:demand}, all paths in $\mathcal{Q}_A^G$ (resp. $\mathcal{Q}_B^G$) are edge-disjoint. Since $H$ is a contracted graph of $G$, all paths in $\mathcal{Q}_A^H$ (resp. $\mathcal{Q}_B^H$) are edge-disjoint.
By \Cref{thm:algo-property}, all paths in $\mathcal{Q}$ are edge-disjoint. Thus for every edge $e\in E(H)$, there exist at most one path in $\mathcal{Q}_A^H$ containing $e$, at most one path in $\mathcal{Q}_B^H$ containing $e$ and at most one path in $\mathcal{Q}$ containing $e$. Therefore, there are at most three paths in $\mathcal{Q}_T$ containing $e$ which implies that $\congestion_H(F')\leq 3$.
\end{proof}

\subsection{Proof of \Cref{thm:edge-disjoint-path}}\label{sec:pf edge disjoint path}

In this section, we give the proof of \Cref{thm:edge-disjoint-path}. Let $G=(V,E,c)$ be the graph constructed in \Cref{subsec:edge cap} and $T$ be the terminal set of size $k$. We first define a
subgraph $\hat{G}=(V,\hat{E},\hat{c})$ of $G$ as follows.
\begin{itemize}
    \item $\hat{E}=E\setminus E'$ where $E'= E(G[V_T])\setminus E_T$. In other words, $\hat{G}$ is obtained from $G$ by removing all edges in induced subgraph $G[V_T]$ except edges in $E_T$.
    \item  For edge capacities, $\hat{c}(e)=c'(e)\leq c(e)$ if $e\in E_T$ and $\hat{c}(e)=c(e)=1$ for other cases. Note that the subgraph $\hat{G}[V_T]$ of $\hat{G}$ induced by $V_T$ satisfies $\hat{G}[V_T]=G_T=(V_T,E_T,c')$.
\end{itemize}
Since $\hat{G}\subseteq G$ is a subgraph, any cut $S\subseteq V$ has $c(\delta_G(S))\geq \hat{c}(\delta_{\hat{G}}(S))$.

\begin{lemma}\label{lem:S in hat{G}} 
For any subset $S\subseteq V$ containing $T$ of size $|S|\leq 0.49n$, $\hat{c}(\delta_{\hat{G}}(S))\geq \hat{c}(\delta_{\hat{G}}(S\cup V_T))$.
\end{lemma}
\begin{proof}
    If $S$ contains $V_T$, it is correct trivially.
    Suppose $V_T\setminus S\neq \emptyset$. Set up $S^\dagger=S\cup V_T$. 
    In the following, we prove 
    $\hat{c}(\delta_{\hat{G}}(S))\geq \hat{c}(\delta_{\hat{G}}(S^\dagger))$ by induction. Let $S^\dagger=S\cup V_T$ be constructed in the following way: 
    \begin{enumerate}
        \item Set $S'\leftarrow S$ initially.
        \item Go through each $i\in \{1,2,\ldots, r\}$,
        \begin{enumerate}
            \item For each node $u\in N_{i}\setminus S'$, add $u$ into $S'$ (e.g. $S'\leftarrow S'\cup \{u\}$).
        \end{enumerate}
    \end{enumerate}
    Finally we achieve that $S'=S\cup V_T=S^\dagger$. 
    It suffices to show that each time a node is added into $S'$, the cut value does not increase. Consider any iteration $1\le i^*\le r$. Assume $N_{j}\subseteq S'$ for any $j< i^*$ and there exists node $u\in V_{i^*}\setminus S'$. Let $S''=S'\cup \{u\}$. 
    From the construction of $\hat{G}$, for each node $u\in N_{i}$, 
    \begin{itemize}
        \item $u$ has exactly one incident edge $e_u=(u,u')$ with $u'\in N_{i-1}$, $e_u\in \delta_{\hat{G}}(S')$, $e_u\notin \delta_{\hat{G}}(S'')$; and
        \item if we denote $E_u=\{(u,u'')\in E_T:u''\in N_{i+1}\}$ as the set of all other incident edges, then each $e'\in E_u$ is either $e'\in \delta_{\hat{G}}(S')$ or $e'\in \delta_{\hat{G}}(S'')$ and not both at the same time. 
    \end{itemize}
    Therefore, by the definition of edge capacities,
    \[
    \hat{c}(\delta_{\hat{G}}(S'))-\hat{c}(\delta_{\hat{G}}(S''))\geq \hat{c}(e_u)-\hat{c}(E_u)=c'(e_u)-c'(E_u)=0.
    \]
    Similarly, for each added node, we can prove that $\hat{c}(\delta_{\hat{G}}(S))\geq \hat{c}(\delta_{\hat{G}}(S^\dagger))$ where $S^\dagger=S\cup V_T$.
\end{proof}
\begin{lemma}\label{lem:cut-lb}
    Given graph $G=(V,E,c)$ and terminal set $T$ of size $k$, for any subset $S\subseteq V$ containing $T$ of size $|S|\leq 0.49n$, its cut value $c(\delta_G(S))=\sum_{e\in \delta_G(S)}c(e) \geq m=10n/(\log n)^{\alpha}$.
\end{lemma}
\begin{proof}
    Consider any subset $S$ containing $T$ of size $|S|\leq 0.49n$.
    Suppose $S$ contains $V_T$, thus 
    $|S|\geq |V_T|=|B_r|\geq 0.2m$. 
    Since $\Phi(G)\geq \frac{1}{2}$ and degree $d=10$, $|\delta_G(S)|\geq 5|S|\geq m$. Since $c(e)=1$ for each $e\in \delta_G(S)$, $c(\delta_G(S)) \geq |\delta_G(S)|\geq m$.
    
    Consider the case where $V_T\setminus S\neq \emptyset$, set $S^\dagger=S\cup V_T$. We work on the subgraph $\hat{G}=(V,\hat{E},\hat{c})$ of $G$ which satisties that $c(\delta_G(S))\geq \hat{c}(\delta_{\hat{G}}(S))$.
    By \Cref{lem:S in hat{G}}, $\hat{c}(\delta_{\hat{G}}(S))\geq  \hat{c}(\delta_{\hat{G}}(S^\dagger))$.
    Since $S^\dagger$ contains $V_T$, each edge $e=(u,v)\in \delta_{\hat{G}}(S^\dagger)$ satisfies that $u\notin B_r$ and $v\notin B_r$, which implies that $e\notin E_T$. Therefore $\delta_{\hat{G}}(S^\dagger)=\delta_G(S^\dagger)$. Additionally, since each $e\in \delta_{\hat{G}}(S^\dagger)$ has $\hat{c}(e)=c(e)=1$, $\hat{c}(\delta_{\hat{G}}(S^\dagger))= |\delta_{\hat{G}}(S^\dagger)|=|\delta_G(S^\dagger)|$.
    Given that $0.2m\leq |V_T|\leq |S^\dagger|\leq |S|+|V_T|\leq 0.49n+2m\leq 0.5n$ and $\Phi(G)\geq \frac{1}{2}$, we have
    \[
    c(\delta_G(S))
    \geq \hat{c}(\delta_{\hat{G}}(S))
    \geq \hat{c}(\delta_{\hat{G}}(S^\dagger))
    = |\delta_{\hat{G}}(S^\dagger)|
    = |\delta_G(S^\dagger)|
    \geq 5|S^\dagger|\geq m = \frac{10n}{(\log n)^{\alpha}}.
    \] 
\end{proof}
\begin{proof}[Proof of \Cref{thm:edge-disjoint-path}]
    Given graph $G=(V,E,c)$, terminal set $T$ and set $U=\{u_1,u_2,\ldots, u_m\}$, set up a new graph $G^*=(V^*,E^*,c)$ where $V^*=V\cup \{s^*,t^*\}$ and $E^*=E\cup E^*_s\cup E^*_t$ with $E^*_s=\{(s^*, t):t\in T\}$ and $E^*_t=\{(t^*, u):u\in U\}$. Each edge $e\in E^*_s$ has $c(e)=+\infty$ and each edge $e\in E^*_t$ has $c(e)=1$. For all remaining edge $e\in E$, its capacity $c(e)$ remains the same as in $G$.

    Consider any cut $(S^*,V^*\setminus S^*)$ in $G^*$ with $s\in S^*$ and $t\notin S^*$. It suffices to prove that for any cut $S^*\supseteq T$, its cut value $|\delta_{G^*}(S^*)|\geq m$. Set up $S=S^*\setminus \{s\}$. 
    \begin{itemize}
        \item Suppose $U \cap S = \emptyset$. Thus $\delta_{G^*}(S^*)=\delta_G(S)$ and $c(\delta_{G^*}(S^*))=c(\delta_G(S))$.
        If $S$ has size $|S|\leq 0.49n$, by \Cref{lem:cut-lb}, it has $c(\delta_{G}(S))\geq m$. 
        Suppose $S$ has size $|S|> 0.49n$. Since $U\subseteq V\setminus S$ and $\Phi(G)\geq \frac{1}{2}$, we have
        $c(\delta_{G}(S))\geq |\delta_{G}(S)|\geq 50\min \{|S|,|V\setminus S|\}\geq 50\min \{0.49n, |U|\} = 50m$.
        \item Suppose $U\cap S\neq \emptyset$. Thus we can partition $U=U_s\sqcup U_t$ such that $U_s\subseteq S$ and $U_t\subseteq V\setminus S$, so $c(\delta_{G^*}(S^*))=\sum_{u\in U_s}c(u,t^*) + c(\delta_G(S))=|U_s| + c(\delta_G(S))$. If  $|S|\leq 0.49n$, by \Cref{lem:cut-lb}, it has $c(\delta_{G}(S))\geq m$. 
        Suppose $S$ has size $|S|> 0.49n$. Since $U_t\subseteq V\setminus S$ and $\Phi(G)\geq \frac{1}{2}$, we have
        $c(\delta_{G}(S))\geq |\delta_{G}(S)|\geq 50\min \{|S|,|V\setminus S|\}\geq 50\min \{0.49n, |U_t|\} = 50|U_t|$. Therefore, 
        \[
        c(\delta_{G^*}(S^*))=|U_s| + c(\delta_G(S)) \geq |U_s| + 50|U_t| \geq |U|=  m.
        \]
    \end{itemize}
    Therefore, there always exist a flow in $G^*$ from $U$ to $T$ with value at least $m$ and we can obtain $m$ edge-disjoint paths in $G$ each connecting a distinct vertex of $U$ to some terminal.
    \end{proof}

\section{Proof of \Cref{thm:algo-property}}\label{sec:pf algo-property}
In this section, we give a formal proof of the following lemma:
\algoproperty*
Let $G=(V,E)$ be a graph generated by random graph $G'$. Note that $G$ is a regular graph with degree $d=10$.
Given any parameter $\varepsilon\in (0,1)$, we have $\alpha=\frac{\log 5}{\log 5+1-\varepsilon}$ and  $\beta=\frac{\log 5-0.5\varepsilon}{\log 5+1-\varepsilon}=\alpha-\frac{0.5}{\log 5+1-\varepsilon}\cdot \varepsilon$.
Consider any valid partition $\mathcal{F}$ of $V$ with size $|\mathcal{F}|=f=O(\frac{n}{2^{(\log n)^\beta}})$ such that distinct terminals of $T$ belong to different sets in $\mathcal{F}$. Let $H$ be the graph contracted by $\mathcal{F}$. 

Set up two parameters $\eta$ and $\lambda$ such that $\eta=1-\varepsilon/2<1$ and $2-\lambda = 2^\eta = 2^{1-\varepsilon/2}<2$. Set $s=2^{(\log n)^{\beta}}$ and 
\begin{equation}\label{eq:def of l} 
    l
    = \left\lfloor \frac{(1-\beta)\log \log n -7}{1-\varepsilon/2}\right\rfloor
    = 
    \left\lfloor \frac{\log\log n}{\log 5 + 1 -\varepsilon} - \frac{7}{1-\varepsilon/2}\right\rfloor.
\end{equation}
For each $i\in \{1,2,\ldots, l\}$, set $s_i=s^{(2-\lambda)^i}=s^{2^{\eta \cdot i}}$. Let $s^*=s_l=s^{(2-\lambda)^l}$. 
By the definition of $l$, we have
\begin{equation}\label{eq:5^l}
    5^l\leq 5^A=2^{\alpha \log \log n - 4\log 10}=(\log n)^\alpha/10^4\leq (\log n)^\alpha/2000
\end{equation}
\begin{equation}\label{eq:s^* ub}
    s^*=s^{(2-\lambda)^{l}}\leq s^{(2-\lambda)^{B}}=2^{(\log n)^{\beta}\cdot 2^{\eta B}}
    =\left(2^{(\log n)^{\beta}\cdot 2^{(1-\beta)\log \log n}}\right)^{1/128}
    =n^{1/128}\leq 
0.04n^{0.01}
\end{equation}
\begin{equation}\label{eq:s^* lb}
      s^*=s^{(2-\lambda)^{l}}\geq s^{(2-\lambda)^{B-1}} = 
    (s^{(2-\lambda)^{B}})^{\frac{1}{2-\lambda}} \geq (n^{\frac{1}{128}})^{\frac{1}{2}}\geq n^{1/256}.
\end{equation}
where $A= \frac{\alpha\log\log n-4\log 10}{\log 5}$, $B= \frac{(1-\beta)\log\log n-7}{1-\varepsilon/2}$ and $B-1\leq l\leq B\leq A$ since $\frac{\alpha}{\log 5}=\frac{1-\beta}{1-\varepsilon/2}=\frac{1}{\log 5 + 1 -\varepsilon}$.

\begin{definition}
    Given $u,v\in V$, we say $(u,v)$ is a typical node pair if
    $\dist_G(u,u')\geq \frac{1}{128}(\log n)$ and $\dist_H(u,u')\leq (\log n)^\alpha
    /2000$, where graph $H$ is contracted by partition $\mathcal{F}$.
\end{definition}

\textbf{Organization of this section.} The remainder of this section is dedicated to the proof of \Cref{thm:algo-property} by utilizing typical node pairs.
In \Cref{subsec:find one typical pair}, we design a clustering algorithm which returns a typical node pair with high probability.
Then in \Cref{subsec:find all typical pairs}, building on the algorithm in \Cref{subsec:find one typical pair}, we design a construction algorithm by selecting typical pairs sequentially and complete the proof of \Cref{thm:algo-property}.
 We finish the remaining proofs in \Cref{subsec:remaining pf}.

\subsection{Finding one typical pair}\label{subsec:find one typical pair}

Our algorithm perform $m=O(\frac{n}{2^{(\log n)^\beta}})$ iterations, where in each iteration we find one typical pair $u,v$ together with its $G$-path $Q_G(u,v)$ and $H$-path $Q_H(u,v)$, and then we remove the edges of $Q_H(u,v)$ and all their endpoints from $G$, and the next iteration is continued on the remaining graphs $G$ and $H$. For those endpoints removed in previous iterations, we call them \emph{useless vertices}.

\subsubsection{Friendship between clusters}
Fix one valid partition $\mathcal{F}$ of $V$.
Given any useless subset $V_U\subseteq V$ of size $|V_U|\leq n/200$, we say that a pair $F,F'\in \mathcal{F}$ of clusters are \textit{friendly} (denoted as $F\sim F'$) if $E_G(F\setminus V_U,F'\setminus V_U)\neq \emptyset$. Set $F\sim F$ by default. For each cluster $F\in \mathcal{F}$, set up
\begin{itemize}
    \item friend set $N^{+}(F)=\{F'\in \mathcal{F}: F'\sim F\}$,
    \item friendly node set $U^+(F)=\bigcup_{F'\in \mathcal{F}: F'\sim F} F'$,
    \item large friend set $N_*^+(F)=\{F'\in \mathcal{F}:F'\sim F, |F'|\geq s/200\}$,
    \item large friendly node set $U_*^+(F)=\bigcup_{F'\in N_*^+(F)}F'$.
\end{itemize}
For each node $u\in V$,
    set up $U^+(u)=U^+(F(u))$, $N^+(u)=N^+(F(u))$, $U^+_*(u)=U^+_*(F(u))$ and $N^+_*(u)=N^+_*(F(u))$.

\begin{definition}
    Given $G=(V,E)$ generated by $G'$, useless node set $V_U\subseteq V$ and valid partition $\mathcal{F}$, we say
    a node $u\in V$ is \textit{good-0} if its friendly node set satisfies that $|U^+(u)|> s_1$. Otherwise, we say $u$ is \textit{bad-0}.
\end{definition}
We can show the following lemma and its proof is deferred to \Cref{subsec:pf bad-0 nodes}.
\begin{lemma}\label{lem:bad-0_upperbound}
    Given parameter $\lambda\in (0,1)$, with probability $1-o(1)$, for any useless subset $V_U$ with size $|V_U|\leq n/200$ and any valid partition $\mathcal{F}$ of $V$, the total number of \textit{bad-0} nodes in $G$ is at most $0.24n$. 
\end{lemma}
Fix an arbitrary subset $\mathcal{F'}\subseteq \mathcal{F}$. Given any partition $\mathcal{H}'$ of $\mathcal{F'}$, we say $\mathcal{H}'$ is valid-i for some $i\in \{1,2,\ldots l\}$, if each cluster $H\in \mathcal{H}'$ satisfies that $|V(H)|=|\bigcup_{F\in H} F|> s_i$. 
Given partition $\mathcal{H}'$ of $\mathcal{F}'$, we say that a pair $H,H'\in\mathcal{H}'$ of clusters are \textit{friendly} (denoted as $H\sim H'$) if $E_G(V(H)\setminus V_U, V(H')\setminus V_U)$, e.g. there exist two nodes $u\in V(H)\setminus V_U, u'\in V(H')\setminus V_U$ such that $(u,u')\in E$. We set $H\sim H$ by default. 

Given a sequence of subsets $\{\mathcal{F}^1,\ldots, \mathcal{F}^l\}$ of $\mathcal{F}$ and a sequence of partitions 
$\{\mathcal{H}^1,\ldots, \mathcal{H}^l\}$ such that $\mathcal{H}^{i}$ is a valid-i partition of $\mathcal{F}^i$ for each $i\in \{1,2,\ldots l\}$. Go through each partition $\mathcal{H}^i$: 
\begin{itemize}
    \item For each $F\in \mathcal{F}^i$, let $H^i(F)\in \mathcal{H}^i$ be the cluster in $\mathcal{H}^i$ containing $F$. For each $u\in V(\mathcal{F}^i)$, let $H^i(u)=H^i(F(u))$.
    \item For each $H\in \mathcal{H}^i$, set up friend set $N^{+}_i(H)=\{H'\in \mathcal{H}^i: H'\sim H\}$
    and friendly node set $U^{+}_i(H)=\bigcup_{H'\in N^+_i(H)} V(H')$.
    \item For each $u\in V(\mathcal{F}^i)$,
    set up $U^+_i(u)=U^+_i(H^i(u))$ and $N_i^+(u)=N_i^+(H^i(u))$.
\end{itemize}
For convenience, set $U_0^+(u)=U^+(u)$.

\begin{definition}
    Given $G=(V,E)$ generated by $G'$, useless node set $V_U\subseteq V$ and valid partition $\mathcal{F}$, go through each $i\in \{1,2,\ldots l\}$. Given any subset $\mathcal{F}^i\subseteq \mathcal{F}$ and valid-i partition $\mathcal{H}^i$ of $\mathcal{F}^i$, we say a node $u\in V(\mathcal{F}^i)$ is good-$i$ if its friendly node set satisfies that $|U^+_i(u)|>s_{i+1}$. Otherwise, we say $u$ is \textit{bad-$i$}. 
\end{definition}

\subsubsection{Clustering algorithm}\label{subsec:clusteringalgo}
We are now ready to describe the clustering algorithm. 

\begin{description}\label{xxx}
    \item[Clustering at level one:] Given  $G=(V,E)$, useless subset $V_U$ and valid partition $\mathcal{F}$ of $V$,
        \begin{enumerate}
    \item Set up an empty set $\mathcal{H}^{1}=\emptyset$. Initially all clusters in $\mathcal{F}$ are unmarked.
    \item Pick an arbitrary good-0 and unmarked cluster $F\in \mathcal{F}$. If all its neighbors $F'\sim F$ are unmarked, set up subset $K=N^+(F)\subseteq \mathcal{F}$, add it into $\mathcal{H}^{1}$ and mark all clusters in $N^+(F)$. Repeat this step until all good-1 and unmarked clusters have at least one neighbor marked.
    \item For each good-0 and unmarked cluster $F$, it has at least one marked neighbor $F'\sim F$ such that there exists $K\in \mathcal{H}^1$ containing $F'$. Add $F$ into $K$.
    \item Let $C^0\subseteq \mathcal{F}$ be the set of all remaining unmarked clusters, which are bad-0. Set up $\mathcal{F}^{1}=\mathcal{F}\setminus C^0$.
\end{enumerate}
    The algorithm ends up with subset $C^0\subseteq \mathcal{F}$ and partition $\mathcal{H}^1$ of $\mathcal{F}^1=\mathcal{F}\setminus C^0$.
    \item[Clustering at higher levels:] For each integer $i\in \{1,2,\ldots l-1\}$, given subset $\mathcal{F}^{i}$ and partition $\mathcal{H}^{i}$ of $\mathcal{F}^i$, 
        \begin{enumerate}
    \item Set up an empty set $\mathcal{H}^{i+1}=\emptyset$. Initially all clusters in $\mathcal{H}^i$ are unmarked.
    \item Pick an arbitrary good-i and unmarked cluster $H\in \mathcal{H}^i$. If all its neighbors $H'\sim H$ are unmarked, set up subset $K= \bigcup_{H'\in N^+_i(H)}H'\subseteq \mathcal{F}^i$, add it into $\mathcal{H}^{i+1}$ and mark all clusters in $N_i^+(H)$. Repeat this step until all good-i and unmarked clusters have at least one neighbor marked.
    \item For each good-i and unmarked cluster $H$, it has a marked neighbor $H'\sim H$ such that there exists $K\in \mathcal{H}^{i+1}$ containing $H'\subseteq K$. Add all clusters of $H$ into $K$.
    \item Let $C^i\subseteq \mathcal{F}^i$ be the union of remaining unmarked clusters, which are bad-i. Set up $\mathcal{F}^{i+1}=\mathcal{F}^i\setminus C^i$.
\end{enumerate}
\end{description}
Given $G=(V,E)$, useless subset $V_U$ of size $|V_U|\leq n/200$ and valid partition $\mathcal{F}$ of $V$, let $\mathcal{A}^*=(C,\mathcal{H})$ be the result of above clustering algorithm with $C=\{C_0,C_1,\ldots C_{l-1}\}$ and $\mathcal{H}=\{\mathcal{H}^1,\mathcal{H}^2,\ldots \mathcal{H}^l\}$. For each integer $i\in \{0,1,2\ldots l-1\}$, $C^i\subseteq \mathcal{F}^i$ is a subset and $\mathcal{H}^{i+1}$ is a valid-i partition of $\mathcal{F}^{i+1}=\mathcal{F}^i\setminus C^i$. Note that $\{C^0, C^1,\ldots C^{l-1},\mathcal{F}^{l}\}$ form a partition of $\mathcal{F}$.

For convenience, we say a node $u\in V$ is bad if $u\in V(C^j)$ for some $j\in \{0,1,\ldots l-1\}$. Otherwise we say $u$ is good. For each $F\in \mathcal{F}$, we say $F$ is bad-j/bad if $F\in C^j$ for some $j\in \{0,1,\ldots l-1\}$.
Let $b_i=|V(C_i)|$ be the number of nodes in $V(C_i)$. We have the following theorem and the proof of \Cref{thm:result good whp} is deferred in \Cref{subsec:pf result good whp}.

\begin{restatable}{lem}{clusteringgoodwhp}\label{thm:result good whp}
    With probability $1-o(1)$, 
    for any useless subset $V_U\subseteq V$ with size $|V_U|\leq n/200$ and any valid partition $\mathcal{F}$ of $V$, the clustering result $\mathcal{A}^*$ of $G$ has at most $0.96n$ bad nodes, i.e. $\sum_{i=0}^{l-1}b_i \leq 0.96n$. 
\end{restatable}

\begin{figure}
    \centering
    \includegraphics[width=0.8\linewidth]{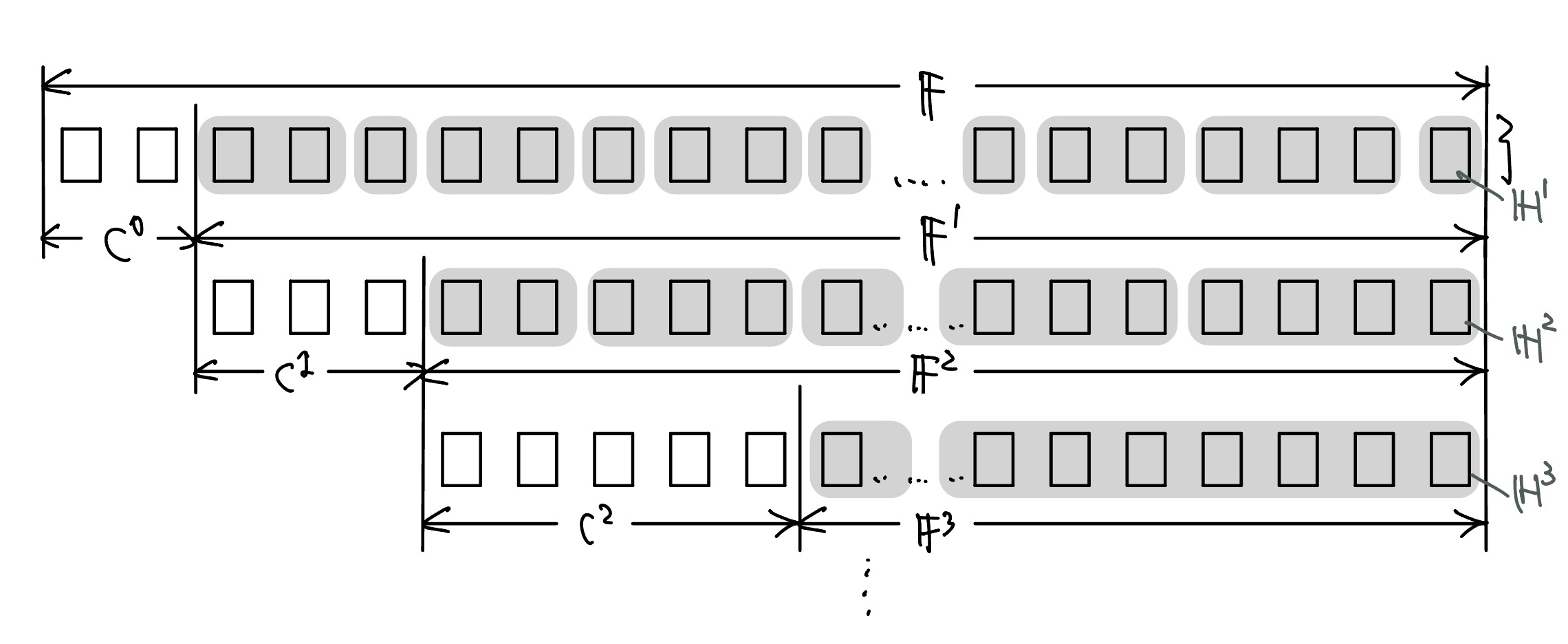}
    \caption{Proper partition result $\mathcal{A}=(C,\mathcal{H})=(C^0,\ldots C^{l-1}, \mathcal{H}^1,\ldots, \mathcal{H}^{l})$.}
    \label{fig: proper partition result}
\end{figure}
\subsubsection{Existence of typical node pairs}
From now on, we assume that the total number of bad nodes $\sum_{i=0}^{l-1}b_i \leq 0.96n$. In the following, we will show that there always exist two non-useless nodes which form a typical node pair. Firstly, we explore the properties of partition $\mathcal{H}_i$.

Let $\mathcal{H}^0=\mathcal{F}$. Go through each integer $i\in \{0,1,\ldots l\}$, given the partition $\mathcal{H}^i$, define the contracted graph $H^i=(\mathcal{H}^i,E^i,c)$ such that for any two distinct cluster $H,H'\in \mathcal{H}^i$, there exists edge $(H,H')\in E^i$ if they are friends, i.e. $H\sim H'$.
{Set $c(H,H')=|E_G(V(H)\setminus V_U,V(H')\setminus V_U)|\geq 1$}. Let $\dist_i(H,H')$ denote the shortest-path distance between two clusters $H,H'\in \mathcal{H}^i$ in graph $H^i$.
\begin{lemma}\label{lem:dist_H^l}
    Given clustering result $\mathcal{A}^*=(C,\mathcal{H})$, for each integer $i\in \{1,2,\ldots, l\}$, the partition $\mathcal{H}^i$ satisfies that for each cluster $H\in \mathcal{H}^i$,
    \[
    \dist_0(F,F')\leq 5^i-1 \quad \text{for any clusters}\quad F,F'\in H.
    \]
\end{lemma}
\begin{proof} We prove by induction. For each cluster $H\in \mathcal{H}^1$, by the clustering algorithm in \Cref{subsec:clusteringalgo}, $H$ contains $N^+(F)$ for some good-0 cluster $F\in \mathcal{F}$. For each $F'\in H\setminus N^+(F)$, there exists $F''\in N^+(F)$ such that $F'\sim F''$. Therefore $\dist_0(F',F'')\leq 4 = 5-1$ for any two cluster $F',F''\in H$. 

Consider each integer $i\in \{2,\ldots ,l\}$ and assume the statement is true for the case $i-1$.
For cluster $K\in \mathcal{H}^i$, by the clustering algorithm in \Cref{subsec:clusteringalgo}, $K$ is constructed from $N^+_{i-1}(H)$ for some good-(i-1) cluster $H\in \mathcal{H}^{i-1}$. Specifically, $K$ contains $H'$ for each $H'\in N_{i-1}^+(H)$.  
For each $H''\in  H^{i-1}$ with
$H''\subseteq K$ and $H''\notin  N_{i-1}^+(H)$, there exists $H'\in N_{i-1}^+(H)$ such that $H'\sim H''$. Therefore, 
\begin{equation}\label{eq:length4_contracted}
    \dist_{i-1}(H',H'')\leq 4 \quad \text{for any cluster} \quad H',H''\subseteq K.
\end{equation}
By induction hypothesis, for each $H'\in \mathcal{H}^{i-1}$ with $H'\subseteq K$, 
\[
\dist_0(F',F'')\leq  5^{i-1}-1 \quad \text{for any cluster} \quad F',F''\in H'.
\]
Consider any two cluster $F',F''\in K$ with $F'\in H'\subseteq K$ and $F''\in H''\subseteq K$ for some $H',H''\in \mathcal{H}^{i-1}$. We can construct a path between $F'$ and $F''$ in $H^0$ as follows:
\begin{enumerate}
    \item Set up a path $P$ between $H'$ and $H''$ in $H^{i-1}$ of length at most $4$, which always exists by Inequality \ref{eq:length4_contracted}. 
    Assume $P=(H_1(=H'),H_2,H_3,H_4,H_5(=H''))$ (w.l.o.g).
    \item For each edge $(H_{j},H_{j+1})$ in path $P$ with $j\in \{1,2,3,4\}$, pick an arbitrary edge $(u_j,v_j)\in E_G(V(H_j)\setminus V_U,V(H_{j+1})\setminus V_U)$ in $G$. Let $u_j\in V(H_j)\setminus V_U$ and $v_j\in V(H_{j+1})\setminus V_U$.
    \begin{itemize}
        \item For cluster $H_1=H'$, set up a 
        shortest path $P'_{0,1}$ in $H^0$ between $F'$ and $F(u_1)$. 
        \item For cluster $H_5=H''$, set up a shortest path $P'_{4,5}$ in $H^0$ between $F(v_4)$ and $F''$. 
        \item For each cluster $H_j$ with $j\in \{2,3,4\}$, it contains $v_{j-1},u_j\in V(H_j)$. Set up a shortest path $P'_{{j-1},j}$ in $H^0$ between $F(v_{j-1})$ and $F(u_j)$. 
    \end{itemize}
    By induction hypothesis, each path $P'_{j,j+1}$ has length at most $5^{i-1}-1$ for each $j\in \{0,1,2,3,4\}$. Additionally, they are edge-disjoint in $H^0$.
    \item For two paths $P'_{j-1,j}$ and $P'_{j,j+1}$ with $j\in \{1,2,3,4\}$, concatenate them using the edge $(F(u_j),F(v_j))$.
    let $P'$ be the final path between $F'$ and $F''$ in graph $H^0$ and $P'$ has length at most $5*(5^{i-1}-1)+4=5^{i}-1$.
\end{enumerate}
Therefore we finish the proof that $\dist_0(F',F'')\leq  5^{i}-1$ for any two cluster $F',F''\in K$.
\end{proof}
\begin{corollary}\label{cor:dist_H^l}
    Given clustering result $\mathcal{A}^*=(C,\mathcal{H})$ with $\sum_{i=0}^{l-1}b_i\leq 0.96n$, it has $|V(\mathcal{F}^l)|\geq 0.04n$ and $\mathcal{H}^l$ is a valid-l partition of $\mathcal{F}^l$. By Inequality \ref{eq:5^l}, each cluster $H\in \mathcal{H}^l$ satisfies that
    \[
    \dist_0(F(u),F(u'))\leq 5^{l} \leq (\log n)^\alpha/2000\quad \text{for any nodes}\quad u,u'\in V(H).
    \]
    
\end{corollary}

\begin{restatable}{lem}{twonodeexists}\label{thm:two-node}
    Given graph $G=(V,E)$, any useless subset $V_U\subseteq V$ with size $|V_U|\leq n/200$, any valid partition $\mathcal{F}$ of $V$ and the clustering result $\mathcal{A}^*$ with $\sum_{i=0}^{l-1}b_i \leq 0.96n$, there exist two non-useless nodes $u,u'\in V\setminus V_U$ such that $(u,u')$ is a typical node pair.
\end{restatable}

\begin{proof}
    Given result $\mathcal{A}^*=(C,\mathcal{H})$ with $\sum_{i=0}^{l-1} b_i\leq 0.96n$, set up $V_G=V(\mathcal{F}^l)$ as the set of all good nodes and $V'_G=V_G\setminus V_U$. 
    Since $|V_U|\leq 0.005n$ and $|V_G|\geq 0.04n$, we have $|V'_G|/|V_G|\geq \frac{7}{8}$.
    For each cluster $H\in \mathcal{H}^l$, set up $V'(H)=V(H)\setminus V_U$. Since $|V'_G|=\sum_{H\in \mathcal{H}^l} |V'(H)|$ and $|V_G|=\sum_{H\in \mathcal{H}^l} |V(H)|$, there exists at least one cluster $H^*\in \mathcal{H}^l$ such that $|V'(H^*)|/|V(H^*)|\geq |V'_G|/|V_G|\geq \frac{7}{8}$. 

    Pick any node $u\in V'(H^*)$, there exists at least one node $u'\in V'(H^*)$ such that $\dist_G(u,u')\geq \lfloor\log_d|V'(H^*)|\rfloor - 1\geq \frac{1}{4}\log |V'(H^*)|$ since $d=10$. Given that $|V'(H^*)|\geq \frac{7}{8}|V(H^*)|$ and $|V(H^*)|\geq s_l = s^*$ since $\mathcal{H}^l$ is valid-$l$ partition, together with Inequality \ref{eq:s^* lb},  
    \[
    \dist_G(u,u')\geq \frac{1}{5}\log(s^*)\geq \frac{\log n}{1280}.
    \]
    On the other hand, since graph $H$ is contracted by $\mathcal{F}$ with $V_U=\emptyset$, $H^0$ is a subgraph of $H$. Together with \Cref{cor:dist_H^l},
    \[
    \dist_H(u,u')\leq 
    \dist_0(F(u),F(u'))\leq 5^{l} \leq (\log n)^{\alpha}/2000.
    \]
    Therefore there exists two non-useless nodes $u, u'\in V\setminus V_U$ such that $(u,u')$ is a typical node pair.
\end{proof}

\subsection{Finding all typical pairs}\label{subsec:find all typical pairs}

Fix graph $G=(V,E)$, any useless subset $V_U\subseteq V$ of size $|V_U|\leq n/200$, any valid partition $\mathcal{F}$ of $V$ and the clustering result $\mathcal{A}^*$ with $\sum_{i=0}^{l-1}b_i \leq 0.96n$.
Suppose $u,u'\in V\setminus V_U$ forms a typical node pair $(u,u')$. Pick one shortest path between $F(u)$ and $F(u')$ in graph $H^0$, denoted by $P^0_H(u,u')$. Now we construct its corresponding edge set $\mathcal{E}(u,u')\subseteq E$ as follows:
\begin{enumerate}
    \item Start from an empty edge set $\mathcal{E}(u,u')=\emptyset$.
    \item For each edge $(F,F')$ in path $Q_H^0(u,u')$, pick some $e\in E_G(F\setminus V_U,F'\setminus V_U)$ and set $\mathcal{E}(u,u')\leftarrow \mathcal{E}(u,u') \cup \{e\}$. 
\end{enumerate}
We say $\mathcal{E}(u,u')$ is one \textit{supporting} edge set of $P^0_H(u,u')$. Note that $$|\mathcal{E}(u,u')|=|P^0_H(u,u')|=\dist_0(F(u),F(u'))\leq (\log n)^{\alpha}/2000.$$ Now we design the construction algorithm and complete the proof of \Cref{thm:algo-property}.

\begin{enumerate}
    \item Set up an empty useless subset $V_U=\emptyset$ and empty collections $\mathcal{P}=\{\}$, $\mathcal{Q}=\{\}$ and $\mathcal{E}=\{\}$.
    \item For each integer $i\in \{1,2,\ldots, m=10n/(\log n)^{\alpha}\}$, apply the clustering algorithm with useless subset $V_U$. Let $\mathcal{A}^*$ be the result. {If $\sum_{i=0}^{l-1}b_i > 0.96n$, the algorithm fails and stops.} Otherwise,
    \begin{itemize}
        \item Pick one typical node pair $(u,u')$ according to \Cref{thm:two-node}. Set  $\mathcal{P}\leftarrow \mathcal{P} \cup \{(u,u')\}$.
        \item Pick one shortest path $P^0_H(u,u')$ in $H^0$ between $F(u)$ and $F(u')$ and set $\mathcal{Q}\leftarrow \mathcal{Q}\cup\{Q_H^0(u,u')\}$.
        \item Construct the supporting edge set $\mathcal{E}(u,u')$ of path $P^0_H(u,u')$. Set $\mathcal{\mathcal{E}} \leftarrow \mathcal{\mathcal{E}} \cup \mathcal{E}(u,u')$.
        \item Set $V_U\leftarrow V_U \cup V(\mathcal{\mathcal{E}}(u,u'))\cup \{u,u'\}$, where $V(\mathcal{\mathcal{E}}(u,u'))$ is the set of endpoints in $\mathcal{\mathcal{E}}(u,u')$.
    \end{itemize}
\end{enumerate}

\begin{proof}[Completing the proof of \Cref{thm:algo-property}]

Assume that given any valid partition $\mathcal{F}$ of $V$ and any useless subset $V_U\subseteq V$ with size $|V_U|\leq n/200$, the clustering result $\mathcal{A}^*$ of $G$ has bad nodes $\sum_{i=0}^{l-1}b_i\leq 0.96n$.
By \Cref{thm:result good whp}, it occurs with probability $1-o(1)$. 

Firstly we claim that under above assumption, the construction algorithm always succeeds.
Suppose the construction algorithm succeeds in the first $t-1$ rounds and fails at iteration $t$ for some $t\in \{1,2,\ldots, m\}$. At the beginning of iteration $t$, $V_U$ has size $|V_U|\leq \frac{10n}{(\log n)^{\alpha}}\cdot (5^l+2)\leq \frac{10n}{(\log n)^{\alpha}}
\cdot (\log n)^\alpha/2000 \leq n/200$ and the clustering algorithm returns result $\mathcal{A}^*$ with bad nodes $\sum_{i=0}^{l-1}>0.96n$. This contradicts to the assumption.

For each iteration $i\in \{1,2,\ldots, m\}$, it produces a pair of nodes $(u_i,u'_i)$ which, by \Cref{thm:two-node}, satisfies that $\dist_G(u_i,u_i')\geq \log n/1280$ and $\dist_0(F(u_i),F(u'_i))\leq (\log n)^\alpha/2000$. Thus the total distance in $G$ is $\sum_{(u,v)\in \mathcal{P}}\dist_G(u,v) \geq \frac{1}{128}n(\log n)^{1-\alpha}$.

Given any iteration $i,j\in \{1,2,\ldots, m\}$ with $i<j$ (w.l.o.g), we have $u_i\neq u_j$ and $u'_i\neq u'_j$ since $u_i,u'_i\in V_U$ at the end of iteration $i$. 
For any $(u_i,u'_i),(u_j,u'_j)\in \mathcal{P}$ with $i< j$, their two selected paths $P^0_H(u_i,u'_i)$ and $P^0_H(u_j,u'_j)$ are edge-disjoint in $H$ since their supporting edge sets $\mathcal{E}(u_i,u'_i)$ and $\mathcal{E}(u_j,u'_j)$ are disjoint. This is achieved since all endpoints of edges in $\mathcal{E}(u_i,u'_i)$ are useless at the end of iteration $i$. Additionally $u_i\neq u_j$ and $u'_i\neq u'_j$ since  $u_i,u'_i$ are useless at the end of iteration $i$.
\end{proof}

\subsection{Missing proofs}\label{subsec:remaining pf}

\subsubsection{Proof of \Cref{lem:bad-0_upperbound}}\label{subsec:pf bad-0 nodes}

\begin{proof}[Proof of Lemma \ref{lem:bad-0_upperbound}]\label{pf:bad-0_upperbound}
Fix any $G=(V,E)$ generated by random graph $G'$ and any useless subset $V_U$ of size $|V_U|\leq n/200$. Note that $G$ is a regular graph with degree $d=10$.
Recall that $s=2^{(\log n)^{\beta}}$ with $\beta=\frac{\log 5-0.5\varepsilon}{\log 5+1-\varepsilon}$. Consider any valid partition $\mathcal{F}$ of $V$ with size $|\mathcal{F}|=f=O(\frac{n}{2^{(\log n)^\beta}})$. We say a cluster $F\in \mathcal{F}$ is \textit{small} if $|F|\leq s/200$. Otherwise we say it is \textit{large}. Let $V_S$ be the set of all nodes in small clusters.    Clearly it has size $|V_S|\leq n/200$. Set up a subset $C\subseteq \mathcal{F}$ which contains bad-0 clusters. Assume there are more than $0.24n$ bad-0 nodes, $|V(C)| > 0.24n$.

\paragraph{Count the valid tuples.}
Given valid $\mathcal{F}$, useless subset $V_U$ and subset $C\subseteq \mathcal{F}$, we say the partial friendship $\sim_C$ on $C$ realizes $C$ if each node $u\in V(C)$ has its large friendly node set of size $|U^+_*(u)|\leq s^{2-\lambda}$. 
Recall that each $F'\in U^+_*(u)$ has size $|F'|\geq s/200$, thus there are at most $|N_*^+(u)|\leq s^{2-\lambda}/(s/200)=200s^{1-\lambda}$ and there are at most $\binom{f}{200s^{1-\lambda}}$ choices of $U^+_*(u)$ such that $|U^+_*(u)|\leq s^{2-\lambda}$. Therefore, given any valid partition $\mathcal{F}$,  useless subset $V_U$ and subset $C\subseteq \mathcal{F}$ with size $|V(C)|>0.24n$, there are at most
    \begin{equation}\label{eq:partial fri on C}
        \binom{f}{200s^{1-\lambda}}^{f}\leq (f)^{200f\cdot s^{1-\lambda}}= (f)^{O(\frac{n}{\lambda})} \leq n^{o(n)}
    \end{equation}
    partial friendships on $C$ such that each node $u\in V(C)$ satisfies that $|U_*^+(u)|\leq s^{2-\lambda}$.

    For convenience, we say a tuple $(\mathcal{F},V_U , C,\sim_{C})$ is valid if $\mathcal{F}$ is a valid partition of $V$,  $V_U$ is a useless subset of size $|V_U|\leq n/200$,
    subset $C\subseteq \mathcal{F}$ has $|V(C)| > 0.24n$ and $\sim_C$ is a partial friendship on $C$ which realizes $C$.
    Since there are at most $n^n$ choices of $\mathcal{F}$, at most $2^n\leq n^{o(n)}$ choices of $V_U$,
    and at most $2^{f}\leq n^{o(n)}$ choices of $C\subseteq \mathcal{F}$. Combined with Inequality \ref{eq:partial fri on C}, the total number of valid tuples is at most $n^n\cdot n^{o(n)}\cdot n^{o(n)} \cdot n^{o(n)}\leq n^{(1+o(1))n}$.

    \paragraph{Compute the probability of each valid tuple.}
    Recall that $G'=(V,\bigcup_{i=1}^{d}M_i)$, where $M_i$ denotes $i$-th perfect matching sampled independently and uniformly. For each perfect matching $M_i$, we view it as constructed in three steps $M_i=M_i^{(1)}\cup M_i^{(2)}\cup M_i^{(3)}$.\footnote{We firstly assume all small nodes and useless nodes are matched properly. In this way, all remaining nodes are in large clusters and we can upper bound all cases. Note that  $|U^+(u)|\leq s^{1.9}$ implies that the case for large case $|U_*^+(u)|\leq s^{1.9}$, the upper bound is proper. 
    } 
    \begin{enumerate}
        \item In the first step, we sample the matching for each node $u\in V_S$ in small clusters and each useless node $u\in V_U$, denoted by $M^{(1)}_i$. Since $|V_S|\leq n/200$ and $|V_U|\leq n/200$, at most $0.02n$ nodes get matched. Note that $M^{(1)}$ realized $\sim_C$ with probability at most $1$.
        \item In the second step, we sample the matching for unmatched nodes in $V(C)$, denoted by $M_i^{(2)}$. We now compute the probability that $M_i^{(2)}$ realizes the partial relationship $\sim_{C}$.
        We pick an unmatched node $u\in V(C)$ sequentially, sample an unmatched node $u'\in U_*^+(u)$ uniformly at random and add edge $(u,u')$ into $M^{(2)}_i$. 
        Since there are at most $0.02n$ nodes matched in the first step, there exist at least $0.98n$ nodes in $V$ unmatched and at least $0.22n$ nodes in $V(C)$ unmatched. Therefore, $M_i^{(2)}$ realizes the partial friendship $\sim_C$ with probability at most:
    \[
    \frac{s^{2-\lambda}}{0.98n}
    \cdot\frac{s^{2-\lambda}}{0.98n-2}
    \cdot\frac{s^{2-\lambda}}{0.98n-4}
    \cdot
    \ldots
    \cdot\frac{s^{2-\lambda}}{0.98n-2\cdot(0.11n-1)}
    \leq 
    \left(\frac{s^{2-\lambda}}{0.9n}\right)^{0.11n}
    \leq 
    \left(n^{-0.99}\right)^{0.11n}.
    \]
    \item In the last step, we sample all remaining edges arbitrarily, denoted by $M_i^{(3)}$.
    \end{enumerate}
    Since each matching is sampled independently and $d=10$, $\bigcup_{i=1}^{d=10} M_i^{(2)}$ realizes the partial friendship $\sim_{C}$
    with probability at most $(n^{-0.99})^{1.1n}\leq n^{-1.08n}$.

    \paragraph{Overall probability.}
    There are at most $n^{(1+o(1))n}$ valid tuples $(\mathcal{F},V_U, C, \sim_{C})$ and each occurs with probability at most $n^{-1.08n}$. Therefore the total probability is at most $n^{-0.08n+o(n)}=o(1)$.  
\end{proof}

\subsubsection{Proof of \Cref{thm:result good whp}}\label{subsec:pf result good whp}
In this section, we give the proof of \Cref{thm:result good whp}. We first define \textit{proper partition result}.

\begin{definition}[Proper partition result]
    Given graph $G=(V,E)$, useless subset $V_U$ and valid partition $\mathcal{F}$ of $V$, we use $\mathcal{A}=(C,\mathcal{H})=(C^0,\ldots C^{l-1}, \mathcal{H}^1,\ldots, \mathcal{H}^{l})$ to denote a proper partition result which satisfies that,
    \begin{itemize}
        \item $\{C^0,C^1,\ldots, C^{l-1}, \mathcal{F}^l\}$ forms a partition of $\mathcal{F}$.
        \item For each $i\in \{1,2\ldots, l\}$, $\mathcal{H}^i$ is a valid-i partition of $\mathcal{F}^i=\mathcal{F}\setminus (C^0 \sqcup\ldots\sqcup C^{i-1})$, which satisfies that $|V(H)| > s_{i}$ for each $H\in \mathcal{H}^i$.
        \item For each $i\in \{0,1,\ldots, l-1\}$, each node $u\in V(C_i)$ is bad-$i$, which satisfies that $|U^+_i(u)|\leq s_{i+1}$.
    \end{itemize}
    Given proper partition result $\mathcal{A}$, we say a node $u\in V$ is bad if $u\in C_j$ for some $j\in \{0,1,\ldots l-1\}$. Otherwise we say $u$ is good. Let $b_i=|C_i|$ be the size of set $C_i$. The total number of bad nodes in $\mathcal{A}$ is $\sum^{l-1}_{i=0}b_i$. Additionally, for each cluster $F\in \mathcal{F}$, we say $F$ is bad-i/bad if $H\in C^i$ for some $i\in \{0,1,\ldots l-1\}$.
\end{definition}

Clearly, any partition result from the algorithm in \Cref{subsec:clusteringalgo} is proper. In order to prove \Cref{thm:result good whp}, we prove the following theorem.

\begin{restatable}{lem}{badnodeub}\label{thm:bad_upperbound}
    With probability $1-o(1)$, the random graph $G'$ satisfies that, 
    for any useless subset $V_U\subseteq V$ with size $|V_U|\leq n/200$,
    any valid partition $\mathcal{F}$ of $V$, 
    and any proper partition result $\mathcal{A}=(C,\mathcal{H})$,
    there are at most $0.96n$ bad nodes in $\mathcal{A}$, i.e. $\sum_{i=0}^{l-1}b_i \leq 0.96n$.
\end{restatable}

\begin{proof}
Fix any $G=(V,E)$ generated by random graph $G'$ and any useless subset $V_U$ of size $|V_U|\leq n/200$. Note that $G$ is a regular graph with degree $d=10$.
Recall that $s=2^{(\log n)^{\beta}}$ with $\beta=\frac{\log 5-0.5\varepsilon}{\log 5+1-\varepsilon}$ and for each $i\in \{1,2,\ldots, l\}$, $s_i=s^{(2-\lambda)^i}\leq s^*$ where $2-\lambda = 2^{\eta}=2^{1-\varepsilon/2}$ and $l$ is defined by \Cref{eq:def of l}.
Consider any valid partition $\mathcal{F}$ of $V$ with size $|\mathcal{F}|=f=O(\frac{n}{2^{(\log n)^\beta}})$. Fix one proper partition result $\mathcal{A}=(C,\mathcal{H})$ with $\sum_{i=0}^{l-1} b^i  \geq 0.96n$. Suppose bad-0 nodes have $b_0 \geq 0.24n$, by \Cref{lem:bad-0_upperbound}, it occurs with probability $o(1)$. Now we assume the number of bad-0 nodes $b_0\leq 0.24n$.
    \paragraph{Count the valid tuples.}
    Given valid partition $\mathcal{F}$, useless subset $V_U\subseteq V$ 
    and proper partition result $\mathcal{A}=(C,\mathcal{H})$, we say the partial friendship $\sim_C$ on $C$ realizes result $\mathcal{A}$ if 
    \begin{itemize}
        \item each node $u\in V(C^0)$ has its large friendly node set of size $|U^+_{*}(u)|\leq s_1=s^{2-\lambda}$.
        \item each node $u\in V(C^i)$ with $i\in \{1,2,\ldots l-1\}$ has its friendly node set of size $|U^{+}_i(u)|\leq s_{i+1}$.
    \end{itemize}
    we now count the total number of possible cases.
    \begin{itemize}
        \item For each node $u\in V(C^0)$, there are at most $\binom{f}{200s^{1-\lambda}}$ choices of $U^+_*(u)$ such that $|U^+_*(u)|\leq s^{2-\lambda}$ by using the same argument as \Cref{pf:bad-0_upperbound}.
        \item For each node $u\in V(C^{i})$ with $i\in \{1,\ldots l-1\}$, there are at most $\binom{f}{s_i}$ choices of $U^+_i(u)$ such that $|U^+_i(u)|\leq s_{i+1}$. Since $\mathcal{H}^i$ is valid, each  cluster $H\in \mathcal{H}^i$ has $|V(H)|\geq s_i$ and $H^i(u)$ has at most $s_{i+1}/s_i \leq s_i$ friends in $\mathcal{H}^i$. Additionally, there are at most $b_i/s_i$ bad-i clusters in $C^i$. 
    \end{itemize}
    Therefore the total number of possible partial friendships $\sim_C$ which realize result $\mathcal{A}$ is at most
    \begin{equation}\label{eq:partial fri A}
    \binom{f}{200s^{1-\lambda}}^{f} \cdot 
    \prod_{i=1}^{l-1}\binom{f}{s_i}^{b_i/s_i} 
    \leq 
    n^{o(n)}\cdot \prod_{i=1}^{l-1} (f)^{b_i} 
    \leq n^{o(n)} \cdot n^{n} = n^{(1+o(1))n}.
    \end{equation}
    For convenience, we say a tuple $(\mathcal{F},V_U,\mathcal{A},\sim_{C})$ is valid if $\mathcal{F}$ is a valid partition of $V$, useless subset $V_U \subseteq V$ has size $|V_U|\leq n/200$, $\mathcal{A}=(C,\mathcal{H})$ is a proper partition result with $\sum_{i=0}^{l-1}b_i \geq 0.96n$ and $b_0\leq 0.24n$ and $\sim_\mathcal{C}$ is a partial friendship on $C$ which realizes result $\mathcal{A}$.

    Since there are at most $n^n$ choices of $\mathcal{F}$, at most $2^n\leq n^{o(n)}$ choices of $V_U$ and at most $((f)^{f})^l\leq n^{f\cdot l}\le n^{o(n)}$ choices of proper partition result $\mathcal{A}$ with $\sum_{i=0}^{l-1}b_i \geq 0.96n$ and $b_0\leq 0.24n$. Combined with Inequality \ref{eq:partial fri A}, the total number of valid tuples $(\mathcal{F},V_U, \mathcal{A},\sim_{C})$ is at most $n^n\cdot n^{o(n)}\cdot n^{o(n)}\cdot n^{(1+o(1))n}=n^{(2+o(1))n}$.

    \paragraph{Compute the probability of each valid tuple.}

    Recall that $G'=(V,\bigcup_{i=1}^{d}M_i)$, where $M_i$ denotes $i$-th perfect matching sampled independently and uniformly. For each perfect matching $M_i$, we view it as constructed in four steps $M_i=M^{(1)}_i\cup M^{(2)}_i\cup M^{(3)}_i\cup M^{(4)}_i$.
    \begin{enumerate}
    \item In the first step, we sample the matching for each node $u\in V_S$ in small clusters and each useless node $u\in V_U$, denoted by $M^{(1)}_i$. Since $|V_S|\leq n/200$ and $|V_U|\leq n/200$, at most $0.02n$ nodes get matched. Note that $M^{(1)}$ realized $\sim_C$ with probability at most $1$.
    \item In the second step, we sample the matching for unmatched nodes in $V(C^0)$, denoted by $M_i^{(2)}$. By assumption $|C^0|=b_0\leq 0.24n$, at most $0.48n$ nodes get sampled. Note that $M^{(2)}_i$ realizes $\sim_C$ with probability at most $1$.
    \item In the third step, we sample the matching for unmatched nodes in $V(C^i)$ with $i=1,2,\ldots l-1$ sequentially, denoted by $M_i^{(3)}$.
    
    Now we compute the probability that $M_i^{(3)}$ realizes $\sim_C$.
    We pick an unmatched node $u\in V(C^i)$, sample an unmatched node $u'\in U^+_i(u)$ ($|U^+_i(u)|\leq s^*$) uniformly at random and add edge $(u,u')$ into $M^{(3)}_i$. 
    Since there are at most $0.02n+0.48n=0.5n$ nodes matched in the first and second step, there exist at least $0.5n$ nodes in $V$ unmatched and at least $0.46n$ nodes in $\bigcup_{i=1}^{l-1} V(C^i)$ unmatched. Therefore, $M_i^{(3)}$ realizes the partial friendship $\sim_C$ with probability at most,
    \[
    \frac{s^{*}}{0.5n}
    \cdot\frac{s^{*}}{0.5n-2}
    \cdot\frac{s^{*}}{0.5n-4}
    \cdot
    \ldots
    \cdot\frac{s^{*}}{0.5n-2\cdot(0.23n-1)}
    \leq \left(\frac{s^{*}}{0.04n}\right)^{0.23n}
    \leq \left(n^{-0.99}\right)^{0.23n}.
    \]
    since $s^*\leq 0.04n^{0.01}$ by Inequality \ref{eq:s^* ub}.
    \item In the last step, we sample all remaining edges arbitrarily, denoted by $M_i^{(4)}$.
    \end{enumerate}
    Since each matching is sampled independently and $d=10$, $\bigcup_{i=1}^{d} M_i^{(3)}$ realizes the partial friendship $\sim_C$
    with probability at most $(n^{-0.99})^{2.3n}\leq n^{-2.27n}$.
    \paragraph{Overall probability.}
    There are at most $n^{(2+o(1))n}$ valid tuples $(\mathcal{F},\mathcal{A},\sim_{C})$ and each occurs with probability at most $n^{-2.27n}$. Therefore the total probability is at most $n^{-0.27n+o(n)}=o(1)$. 
\end{proof}

\section{Proof of \Cref{thm: main} for Convex Combinations}
\label{sec: convex combination}

In this section, we extend the lower bound proof to convex combinations of contraction-based flow sparsifiers. Recall that a convex combination of contraction-based flow sparsifiers is a distribution $\mu$ on the set of all contraction-based flow sparsifiers, and graph $H_\mu$ is obtained by gluing the copies of the contraction-based flow sparsifiers in the support of $\mu$ at all terminals, where the edge capacity for the copy of each sparsifier $H$ is scaled by $\Pr_{\mu}[H]$. 

For each sparsifier $H$ in the support of $\mu$, we construct a demand $\mathcal{D}_H$ as in \Cref{subsec:demand}. Define demand $\mathcal{D}_{\mu}=\sum_H \Pr_\mu[H] \cdot \mathcal{D}_H$, we will show that the congestion of $\mathcal{D}_{\mu}$ in $G$ is large, while the congestion of $\mathcal{D}_{\mu}$ in $H_{\mu}$ is small. 

For a flow $F$ in $G$, we define the average congestion $\congestion_G^{avg}(F)$ as follows:
\[
\congestion_G^{avg}(F):=\frac{\sum_{P}f_P \cdot {\rm length}(P)}{\sum_{e\in E} c(e)}
=\frac{\sum _{e\in E}\sum_{P\in \mathcal{P}^G:e\in P}f_P}{\sum_{e\in E} c(e)}.
\]

Similarly, given a demand $\mathcal{D}$ we define the average congestion of $\mathcal{D}$ in $G$, denoted by $\congestion_G^{avg}(\mathcal{D})$, as the minimum average congestion of any flow that routes $\mathcal{D}$ in $G$.

By definition, the average congestion of a demand in $G$ is always at most its congestion in $G$, namely $\congestion_G^{avg}(\mathcal{D})\le \congestion_G(\mathcal{D})$. In \Cref{lem:dist-demand-lb}, for each $\mathcal{D}_H$, what we actually proved is that $\congestion_G^{avg}(\mathcal{D_H})=\Omega\left((\log n)^{1-\alpha}\right)$. This means that, for any flow that routes $\mathcal{D}_H$ in $G$, the sum of flow value on all edges in $G$ is at least $\Omega\left((\log n)^{1-\alpha}\right) \cdot \sum_{e \in E} c(e)$. Since $\mathcal{D}_{\mu}$ is a convex combination of $\mathcal{D}_H$, for any flow that routes $\mathcal{D}_{\mu}$, the sum of the flow on edges are also at least $\Omega\left((\log n)^{1-\alpha}\right) \cdot \sum_{e \in E} c(e)$, which means that $\congestion_G(\mathcal{D}_{\mu}) \ge \congestion_G^{avg}(\mathcal{D}_{\mu})=\Omega\left((\log n)^{1-\alpha}\right)$.

On the other hand, by \Cref{{lem:totalcost-lb}}, for each $H$ and demand $\mathcal{D}_H$, there is a flow $F_H$ that routes $\mathcal{D}_H$ such that the congestion of $F_H$ in $H$ is at most 3. Construct a flow $F_{\mu}$ in graph $H_{\mu}$ as follows, for each $H$, we route $\Pr_{\mu}[H]\cdot F_H$ in the copy of $H$. Since $H_{\mu}$ is obtained by gluing the copies together, the flows do not affect each other, meaning that the congestion of $F_{\mu}$ in $H_{\mu}$ also is at most $3$. 

In summary, any flow that routes $\mathcal{D}_{\mu}$ in $G$ has congestion at least $\Omega\left((\log n)^{1-\alpha}\right)$, and there exists a flow that routes $\mathcal{D}_{\mu}$ in $H_{\mu}$ with congestion at most $3$. Therefore, the quality of $H_{\mu}$ is $\Omega\left((\log n)^{1-\alpha}\right)$.

\bibliographystyle{alpha}
\bibliography{REF}

\end{document}